\newif\ifCONF
\CONFfalse

\newif\ifSUBM
\SUBMfalse

\newif\ifDRAFT
\DRAFTfalse

\ifCONF
\documentclass[10pt, conference, onecolumn, compsocconf]{../IEEE_CS_Latex8.5x11x1/IEEEtran}

\else
\documentclass[11pt]{article}
\usepackage{fullpage}
\fi

\ifSUBM

\DRAFTfalse
\fi

\usepackage{url}
\ifCONF
\usepackage[cmex10]{amsmath}
\else
\usepackage{amsmath}
\fi
\usepackage{mathtools}
\usepackage{amssymb}
\usepackage{graphicx}
\usepackage{stmaryrd}
\usepackage{enumerate}
\usepackage{algorithm}

\ifCONF
\else
\usepackage[colorlinks,urlcolor=blue,citecolor=blue,linkcolor=blue]{hyperref}
\fi

\newcommand{\cE}{\mathcal{E}}

\DeclareMathOperator{\nnz}{\mathtt{nnz}}
\DeclareMathOperator{\E}{\mathbf{E}}
\DeclareMathOperator{\Var}{\mathbf{Var}}
\DeclareMathOperator{\argmin}{argmin}

\DeclareMathOperator{\prob}{\mathbf{P}}
\DeclareMathOperator{\rank}{\mathrm{rank}}
\DeclareMathOperator{\dist}{\mathrm{dist}}
\DeclareMathOperator{\Vol}{\mathrm{Vol}}

\newcommand\cL{{\cal L}}
\newcommand\cN{{\cal N}}
\newcommand\cV{{\cal V}}

\newcommand\cT{{\cal T}}
\newcommand\cS{{\cal S}}
\newcommand\cB{{\cal B}}
\newcommand\cC{{\cal C}}
\newcommand\cM{{\cal M}}
\newcommand\tX{{\tilde X}}

\newcommand\hX{{\hat X}}

\DeclareMathOperator\rspace{\mathrm{rspace}}
\DeclareMathOperator\cspace{\mathrm{cspace}}

\DeclareMathOperator\bone{\mathbf{1}}

\newcommand{\vertiii}[1]{{\left\vert\kern-0.25ex\left\vert\kern-0.25ex\left\vert #1 
    \right\vert\kern-0.25ex\right\vert\kern-0.25ex\right\vert}}
\newcommand{\myvertiii}[1]{{\vert\kern-0.25ex\vert\kern-0.25ex\vert #1 
    \vert\kern-0.25ex\vert\kern-0.25ex\vert}}
    
\newcommand{\norm}[1]{{\| #1 \|}}

\newcommand{\norme}[1]{{\myvertiii{#1}}}

\ifDRAFT
\newcommand{\marrow}{\marginpar[\hfill$\longrightarrow$]{$\longleftarrow$}}
\newcommand{\niceremark}[3]
   {\textcolor{red}{\textsc{#1 #2:} \marrow\textsf{#3}}}
\newcommand{\Ken}[2][says]{\niceremark{Ken}{#1}{#2}}
\newcommand{\David}[2][says]{\niceremark{David}{#1}{#2}}
\usepackage[inline]{showlabels}
\else
\newcommand{\Ken}[1]{}
\newcommand{\David}[1]{}
\fi

\newcommand{\poly}{{\mathrm{poly}}}
\newcommand{\eps}{\varepsilon}
\newcommand{\R}{{\mathbb R}}
\newcommand{\polylog}{{\mathrm{polylog}}}
\newcommand{\Ibr}[1]{\llbracket#1\rrbracket}

\newtheorem{theorem}{Theorem}

\newtheorem{definition}[theorem]{Definition}
\newtheorem{lemma}[theorem]{Lemma}

\newtheorem{fact}[theorem]{Fact}

\newtheorem{remk}[theorem]{Remark}
\newtheorem{exmp}[theorem]{Example}

\newenvironment{remark}{\begin{remk}
\begin{normalfont}}{\end{normalfont}
\end{remk}}

\def\FullBox{\hbox{\vrule width 8pt height 8pt depth 0pt}}

\def\qed{\ifmmode\qquad\FullBox\else{\unskip\nobreak\hfil
\penalty50\hskip1em\null\nobreak\hfil\FullBox
\parfillskip=0pt\finalhyphendemerits=0\endgraf}\fi}

\def\qedsketch{\ifmmode\Box\else{\unskip\nobreak\hfil
\penalty50\hskip1em\null\nobreak\hfil$\Box$
\parfillskip=0pt\finalhyphendemerits=0\endgraf}\fi}

\ifCONF\else
\newenvironment{proof}{\begin{trivlist} \item {\bf Proof:~~}}
  {\qed\end{trivlist}}
\fi

\date{}
\begin{document}


\title{Input Sparsity and Hardness for Robust Subspace Approximation}
\ifCONF
\else 
\author{Kenneth L. Clarkson\\IBM Research, Almaden\\{\tt klclarks@us.ibm.com} \and David P. Woodruff\\IBM Research, Almaden\\{\tt dpwoodru@us.ibm.com}}
\fi
\maketitle
\begin{abstract}
In the subspace approximation problem, we seek a $k$-dimensional subspace $F$ of $\mathbb{R}^d$ that minimizes
the sum of $p$-th powers of Euclidean distances to a given set of $n$ points $a_1, \ldots, a_n \in \mathbb{R}^d$, 
for $p \geq 1$. 
More generally than minimizing $\sum_i \dist(a_i,F)^p$,
we may wish to minimize $\sum_i M(\dist(a_i,F))$ for some loss function $M()$, for example, $M$-Estimators, which
include the Huber and Tukey loss functions. 
Such subspaces provide alternatives
to the singular value decomposition (SVD), which is the $p=2$ case, finding
such an $F$ that minimizes the sum of squares of distances. For $p \in [1,2)$,
and for typical $M$-Estimators, the minimizing $F$ gives a solution that is more robust to
outliers than that provided by the SVD.
We give several algorithmic results for these robust subspace approximation problems.

We state our results as follows, thinking of the $n$ points as forming an $n \times d$ matrix $A$, and 
letting $\nnz(A)$
denote the number of non-zero entries of $A$. 
Our results hold for $p\in [1,2)$.
We use $\poly(n)$ to denote $n^{O(1)}$ as $n\rightarrow\infty$.
\begin{enumerate}
\item For minimizing $\sum_i \dist(a_i,F)^p$, 
we give an algorithm running in
\[
O(\nnz(A) + (n+d)\poly(k/\eps) + \exp(\poly(k/\eps)))
\]
time which outputs a $k$-dimensional 
subspace $F$ whose cost is at most a $(1+\eps)$-factor larger than the optimum.
\item We show that the problem of minimizing $\sum_i \dist(a_i, F)^p$ 
is NP-hard, even to output a $(1+1/\poly(d))$-approximation. This extends work of Deshpande et al. (SODA, 2011) which could only
show NP-hardness or UGC-hardness for $p > 2$; their proofs critically rely on $p > 2$. Our work resolves an open question
of [Kannan Vempala, NOW, 2009]. Thus, there cannot be an algorithm running in time polynomial in $k$ and $1/\eps$ unless P = NP.
Together with prior work, this implies that the problem is NP-hard for all $p \neq 2$. 
\item For loss functions for a wide class of $M$-Estimators, we give a problem-size reduction:
for a parameter $K=(\log n)^{O(\log k)}$, our reduction takes
\[
O(\nnz(A)\log n + (n+d)\poly(K/\eps))
\]
time to reduce the problem to a constrained version involving matrices whose dimensions
are $\poly(K\eps^{-1}\log n)$. We also give bicriteria solutions.
\item Our techniques lead to the first $O(\nnz(A) + \poly(d/\eps))$ time algorithms for $(1+\eps)$-approximate
regression for a wide class of convex $M$-Estimators. 
This improves prior results \cite{cw15}, which were $(1+\eps)$-approximation for Huber regression only,
and $O(1)$-approximation for a general class of $M$-Estimators.
\end{enumerate}
\end{abstract}

\section{Introduction}
In the problem of {\sf Subspace Approximation}, we are given $n$ points $a_1, \ldots, a_n \in \R^d$ and
we want to find a $d\times d$ projection matrix $X$, projecting row vector $a_i$ to $a_iX\in F$, where 
$F$ is a $k$-dimensional subspace,
such that $X$ minimizes $\sum_{i=1}^n M(\|a_i-a_iX\|_2)$, for a given function $M$. The problem fits
in the growing body of work on finding low-dimensional representations of massive data sets,
with applications to clustering, data mining, machine learning, and statistics. 

When $M(x) = x^2$, the problem is principal component analysis (PCA), and the optimal subspace is 
spanned
by the top $k$ right singular vectors of the $n \times d$ matrix $A$ whose rows are the points
$a_1, \ldots, a_n$.
The optimal solution can be computed using the singular value decomposition (SVD) in $\min(nd^2, n^2d)$ time.
By relaxing this to finding a $k$-dimensional subspace with cost at most $(1+\eps)$ times the
optimum, the problem can be solved in $nd \cdot \poly(k/\eps)$ time deterministically \cite{l13,gp14}, where
$\poly(k/\eps)$ denotes a low degree polynomial in $k/\eps$. If a small probability of error is
allowed, the running time can be improved to $O(\nnz(A)) + (n+d)\cdot \poly(k/\eps)$, where $O(\nnz(A))$
denotes the number of non-zero entries of the matrix $A$ \cite{s06,CW13,mm13,nn13,bn13}.
The latter time is useful for sparse matrices, and is optimal in the sense that any algorithm achieving some relative
error with constant probability needs to read $\Omega(\nnz(A))$ entries of $A$. 

The case $M(x) = |x|^p$, $p \geq 1$, was introduced in the theory community by 
Shyamalkumar and Varadarajan \cite{sv07}, and earlier the case $p = 1$ as well as some $M$-Estimators
were studied in the machine learning community by Ding \emph{et al.} \cite{ding2006r}. 
These works 
include the important case of $p = 1$, which provides a more robust solution than the SVD in the
sense that the optimum is less sensitive to outliers. 
Shyamalkumar and
Varadarajan \cite{sv07} give an algorithm for any $p \geq 1$ that runs in time
$nd \cdot\exp((k/\eps)^{O(p)})$, where $\exp(n)$ denotes a function in $2^{\Theta(n)}$. 

Deshpande and
Varadarajan \cite{dv07} refined this, showing that it is possible in $nd \cdot \poly(k/\eps)$ time to
produce a subset of $r = (k/\eps)^{O(p)}$ points, known as a {\it weak coreset},
whose span contains a $k$-dimensional subspace whose cost is at most a factor of $(1+\eps)$ times
the optimal cost. By projecting the $n$ input points onto the span of these $r$ points, one can find
this $k$-dimensional subspace in time exponential in the smaller dimension $r$ using the approach in
\cite{sv07}. The authors thus make the important step of isolating the ``dimension reduction'' step of
the problem from the ``enumeration'' step. This is useful in practice since one can run heuristics
in place of enumeration on the weak coreset, potentially allowing for $k/\eps$ to be much larger while 
still obtaining efficient algorithms \cite{ding2006r,fl11}. 

The time complexity for $p = 1$ was improved
by Feldman et al. \cite{FMSW} to $nd \cdot \poly(k/\eps) + (n+d) \cdot \exp(\poly(k/\eps))$, and
later for general $p$ to $nd \cdot \poly(k/\eps) + \exp((k/\eps)^{O(p)})$ by Feldman and Langberg
\cite{fl11,f14}. The latter work, together with work by Vadarajaran and Xiao \cite{vx12}, also
gives a {\it strong coreset} for \textsf{Subspace Approximation}, i.e., a way of reducing the number of rows
of $A$ so as to obtain a matrix $A'$ so that the cost of fitting the rows of $A'$ to any $k$-dimensional
subspace of $F$ is within a $1+\eps$ factor of the cost of fitting the rows of $A$ to $F$. 

On the hardness side, for constant $p > 2$, Deshpande et al. \cite{dtv11} first give an algorithm
showing it
is possible to obtain a constant factor approximation in $\poly(nd)$ time. They also show 
that for $p > 2$, assuming the Unique Games Conjecture (UGC), that the problem is hard to 
approximate within
the same constant factor, while they show NP-hardness for $p > 2$ to approximate within
a $(1+1/\poly(nd))$-factor. Later, Guruswami et al. \cite{grsw12} show the same constant
factor hardness for $p > 2$ without
the UGC, namely, they show NP-hardness for $p > 2$. 

\subsection{Our Contributions}
Despite the progress on this problem, there are several natural questions that remain open.
On the algorithmic side, a natural question is whether it is possible to obtain a running time
proportional to the number $\nnz(A)$ of non-zero entries of $A$. 
This would match the leading order term in the $p = 2$ case, improving the $nd\cdot \poly(k/\eps)$ leading
order term in the previous works mentioned above (which may be improvable to $O(\nnz(A) \poly(k/\eps))$), and 
join a growing body of
work in numerical linear algebra whose aim is to achieve a running time with leading order term a constant
times the sparsity of the input matrix \cite{bn13,CW13,mm13,nn13,wz13}. 
Our first result is the following. We note that all algorithms mentioned in the following theorems succeed
with constant probability, which can be made arbitrarily small by independent repetition. 

\begin{theorem}[A version of Theorem~\ref{thm Approx Lp}] %
For any $k \geq 1, \eps \in (0,1),$ and $1 \leq p = a/b\in [1,2)$ for integer constants $a,b$,  
there is an $O(\nnz(A)) + (n+d)\poly(k/\eps) + \exp(\poly(k/\eps))$ time algorithm for the
{\sf Subspace Approximation} problem with $M(x) = |x|^p$. 
\end{theorem}
\hspace{5mm} We note that our algorithm is optimal, up to a constant factor, for $k/\eps$ not too large and 
$\nnz(A) \geq (n+d)\poly(k/\eps)$; indeed, in this case the time is $O(\nnz(A))$ and any algorithm
achieving relative error needs to spend $\Omega(\nnz(A))$ time. Moroever, as discussed above, if
one just wants a dimensionality reduction to a set of $\poly(k/\eps)$ points whose span contains
a $k$-dimensional subspace which is a $(1+\eps)$-approximation, then the time is
$O(\nnz(A)) + (n+d)\poly(k/\eps)$, that is, the $\exp(\poly(k/\eps))$ term is removed. This is useful
for large values of $k/\eps$ for which the heuristics mentioned above can be run. 

Another question is whether the $\exp(\poly(k/\eps))$ term in the time complexity is necessary in the
previous theorem. 
All previous algorithms have such a term in their complexity, while known hardness results
apply only for $p > 2$. The need for $p > 2$ is essential in previous hardness results, as the hard instances
in \cite{dtv11} (and similarly \cite{grsw12} which builds upon \cite{dtv11}) become easy for $p < 2$. 
Indeed, the inapproximability ratio shown in these works is $\gamma_p$, the $p$-th moment of a standard
normal distribution, which is less than $1$ for $p \in [1,2)$. We 
note that \cite{dtv11} also shows a weaker NP-hardness but also
only for $p > 2$, and Case 1 in their proof heavily relies on the assumption that $p > 2$. 
In Section 1.4 of the monograph of Kannan
and Vempala, the second open question is whether it is NP-hard to find a subspace of dimension at most
$k$ that minimizes the sum of distances of the points to a subspace, i.e., the $p = 1$ case in our notation.
We resolve this question as follows.

\begin{theorem}[Informal version of Theorem~\ref{thm:main:hardness}]
For any $p \in [1,2)$, 
it is NP-hard to solve the {\sf Subspace Approximation} problem up to a factor of $1+1/\poly(d)$. 
\end{theorem}
\hspace{5mm} Our result, when combined with the hardness results for $p > 2$, shows there is a {\it singularity}
at $p = 2$, namely, for $p = 2$ there is a polynomial time algorithm for any $k, \eps$, while for any other
$p$ the problem is NP-hard. It also shows there cannot be an algorithm running in time polynomial in $k$ and $1/\eps$,
unless P = NP. 

Next, we consider the many other loss functions used in practice, in particular, 
those for $M$-Estimators. 
This has been studied in \cite{fs12} for point and line median, and recently in \cite{cw15} for
regression. Such loss functions include important special cases such as the Huber loss
function, the $\ell_1-\ell_2$ loss, the Tukey function, etc. We refer the reader to \cite{cw15}
for more discussion on these. Many of these loss functions have the property that $M(x) \approx x^2$
for $x$ near the origin, while $M(x) \approx |x|$ for larger $x$. Thus, they enjoy 
the smoothness
properties of $\ell_2^2$ yet also the robustness properties of $\ell_1$. 
As one practical example: in the context of analysis of astronomical spectra, Budavari et al. \cite{BWSDY}
give an algorithm for robust PCA, using an $M$-Estimator in a way quite similar to ours.
One challenge that arises
with $M$-Estimators is that unlike norms they are not scale-invariant, and may have very different
behaviors in different regimes of input values. Prior to this work, to the best of our knowledge no
such results were known in the context of low rank approximation. We give the first algorithm for a
general class of $M$-Estimators for a fixed constant factor approximation; moreover the time complexity is
nearly linear in $\nnz(A)$. We also give two general dimensionality reduction results for general
$(1+\eps)$-approximation, in the spirit of the coreset results stated above.

\begin{definition}[\emph{nice} functions for $M$-estimators, $\cM_2$, $\cL_p$]
We say an $M$-Estimator is {\bf nice} if $M(x) = M(-x), M(0) = 0$, $M$ is non-decreasing in $|x|$,
there is a constant $C_M > 0$ and a constant $p \geq 1$ so that for all $a,b \in \mathbb{R}^+$ with $a \geq b$, we have
$$C_M \frac{|a|}{|b|} \leq \frac{M(a)}{M(b)} \leq \left(\frac{a}{b} \right )^p,$$
and also that  $M(x)^{1/p}$ is subadditive, that is, $M(x+y)^{1/p} \leq M(x)^{1/p} + M(y)^{1/p}$.
 
Let $\cM_2$ denote the set of such nice $M$-estimators, for $p=2$. Let $\cL_p$ denote $M$-estimators
with $M(x)=|x|^p$ and $p\in [1,2)$.
\end{definition}

\begin{remark}
Well-studied $M$-Estimators such as the $L_1-L_2$ loss $M(x) = 2(\sqrt{1+x^2/2}-1)$, 
the Fair estimator loss $M(x) = c^2 \left [\frac{|x|}{c} - \log(1+\frac{|x|}{c}) \right ]$, 
and the Huber loss $M(x) = x^2/(2\tau)$ if $|x| \leq \tau$, and $M(x) = |x|-\tau/2$ otherwise, are all
nice $M$-Estimators. (The proof that $L_1-L_2$ is subadditive requires a calculation%
\ifCONF
\else
; we've included one in Appendix \ref{app:l1l2loss}. %
\fi
)

Here, $c$ and $\tau$ are positive constants. The linear growth lower bound is satisfied
by any convex function $M$, though in general a nice estimator need not be convex. The linear growth lower bound
also rules out \emph{redescending} $M$-estimators, for which $M'(x)\rightarrow 0$ as $|x|\rightarrow\infty$,
but note that we allow $M'(x)$ to decrease, just not all the way to zero. We can allow
$C_M\le M'(x)$ to be arbitrarily small, at a computational cost, so loosely speaking we can
get ``close'' to some redescending $M$-estimators.
\end{remark}

\begin{theorem}[Informal, from Theorems~\ref{thm Approx M2}, \ref{thm const approx}, and \ref{thm non adapt}]
\label{thm:mixedResults}
For any nice $M$-Estimator $M()$ in $\cM_2$, integer $k>0$,
and $\eps \in (0,1)$, for a parameter $K$ in $(\log n)^{O(\log k)}$,
we can in $O(\nnz(A)\log n + (n+d) \poly(K/\eps) )$
time reduce {\sf Subspace Approximation} for $M$
to the problem of solving an instance of $\min_{\rank X=k} \sum_{i\in [n]} M(\norm{\hat A_{i*}XB - C_{i*}}_2)$,
for matrices $\hat A$, $B$, and $C$ with dimensions in $\poly(K\eps^{-1}\log n)$.

In time $O(\nnz(A)+ (n+d)\poly(k))$, we can find a subspace
of dimension $\poly(k\log n)$ whose cost is within $K$ of the best $k$-dimensional subspace.

In time $O(\nnz(A)\log n +  (n+d)\poly(K/\eps) )$, we can find a a subspace of dimension $\poly(K/\eps)$
that contains a $k$-dimensional subspace whose cost is within $1+\eps$ of the best $k$-dimensional subspace.
\end{theorem}
Thus, we make significant progress for nice $M$-estimators for \textsf{Subspace Approximation}.

Finally, using the techniques developed here for $M$-Estimators, we are able to strengthen the results for
$M$-Estimators for the {\sf Regression Problem} in \cite{cw15}, which is the problem of finding an $x \in \mathbb{R}^d$
for which $\|Ax-b\|_M \leq (1+\eps)\min_{x'} \|Ax'-b\|_M$ given an $n \times d$ matrix $A$ and an $n \times 1$ vector $b$.  
Here, for a vector $z \in \mathbb{R}^n$, $\|z\|_M^2 = \sum_{i=1}^n M(z_i)$. In \cite{cw15}, it was shown how to do this
for the Huber loss function in $O(\nnz(A)\log n) + \poly(d \eps^{-1} \log n)$ time, while for nice $M$-Estimators
it was shown how to, in $O(\nnz(A) \log n) + \poly(d \log n)$ time,
obtain a fixed constant-factor approximation via
sketching techniques. We improve upon the latter using sampling-based techniques. 

\begin{theorem}
For any convex $M$-Estimator in $\cM_2$, it is possible in $O(\nnz(A)\log n) + \poly(d/\eps))$ time, to solve
the {\sf Regression Problem} up to a factor of $1+\eps$.
\end{theorem}

In the remainder of this paper, we outline our techniques, first for the hardness result,
and then for the algorithms.

We use the notation $[m]\equiv \{1,2,\ldots, m\}$ for integer $m$.

\subsection{Technical Overview: Hardness}

We first observe that for the simplex $E = \{e_i \mid i \in [d]\}$, the optimal $k$-dimensional
subspace for $M(x) = |x|^p$, for $p \in [1,2)$, is exactly one of the $k$ {\it coordinate spaces}, i.e., a subspace
formed by the span of $k$ standard unit vectors. All such subspaces have the same cost, and correspond to a
subspace with $k$ leverage scores equal to $1$ and remaining leverage scores equal to $0$. In our input, we
include $\poly(d)$ copies of the simplex, which intuitively forces the optimal $k$-dimensional subspace for
our input to be very close to a coordinate space, where we formalize closeness by looking at how close the $(k+1)$-st
leverage score is to $0$. 

As part of our input, we also include $d$ points, which correspond to rows of a $d \times d$ matrix $A$. We create
$A$ from the adjacency matrix of an $r$-regular graph $G$. Namely, for a sufficiently large value $B_1 = \poly(d)$, 
for $i \neq j$, $A_{i,j} = c/\sqrt{B_1r}$ if $\{i,j\}$ is an edge
in $G$, and $A_{i,j} = 0$ otherwise, where $c = \sqrt{2} - O(1/B_1)$ in $(1,2)$. 
Also, for all $i \in [d]$,  $A_{i,i} = 1-1/B_1$. 

The goal is to decide if the maximum clique in $G$ is of size at least $k$ or at most $k-1$. Since we have forced
the optimal $k$-dimensional subspace to be a coordinate subspace, we can think of the $k$ dimensions chosen
as a set $S$ of $k$ vertices in $G$, which correspond to rows of $A$. The contribution of one row $A_i$ of $A$ to
the objective function is $(1-\|A_i^S\|_2^2)^{p/2}$, 
where $A_i^S$ is a vector which agrees with $A_i$ on coordinates in $S$, and is
$0$ on coordinates outside of $S$. If $i \in S$, then we can show 
$(1-\|A_i^S\|_2^2)^{p/2} = (1/B_1^{p/2})(2-2e(i,S)/r - O(1/B_1))^{p/2}$, where $e(i,S)$ is the number of edges
from vertex $i$ to vertices in the set $S \setminus \{i\}$, which can be at most $k-1$. Further,
one can assume $k \leq r$.  
On the other hand, if $i \notin S$, then
$(1-\|A_i^S\|_2^2)^{p/2} = 1- O(1/B_1)$. Since $|S| = k$, the contribution to the objective function
from all $i \notin S$ is $d-k - O(d/B_1)$. Note that the contribution from a single $i \in S$ is
$(1/B_1^{p/2}) (2- 2e(i,S)/r - O(1/B_1))^{p/2}$, and since $p < 2$ and $e(i,S)$ is an integer less than $r$, this is much larger than $O(d/B_1)$
for $B_1 = \poly(d)$ sufficiently large. Therefore, we can think of the contribution from all $i \notin S$
as being the fixed value $d-k$. One can show then if there is a clique of size at least $k$, that the
contribution from all $i \in S$ is a $(1+1/\poly(d))$ factor larger than if the clique size is at most $k-1$. 

In the proof above, we note that the clique size enters as a {\it low order term}, but we are able to fix
the high order terms so we can still extract it with a $(1+1/\poly(d))$-approximation. Finally, we show
that if the subspace is close enough to a coordinate subspace, the analysis above goes through; 
otherwise it is too far from a coordinate subspace, and the cost just on the copies of the
simplex alone is too large. 

\subsection{Technical Overview: Algorithms}

As with many recent papers on randomized numerical linear algebra, we use
a series of randomized matrix techniques, which we generically call \emph{sketching},
to reduce
the original problem to problems involving matrices with fewer rows, or columns, or both.
We extend or speed up these methods.

In the following, we discuss a series of these methods
and the context in which we use them; first, a sketching matrix that reduces the
dimensionality (the number of columns), then sampling that reduces the number of points (rows),
then dimensionality again, then points again, and then the solution of the resulting small optimization problems.
We then discuss, in \S\ref{subsubsec lev score intro}, the
fast estimation of leverage scores, followed by a discussion of the particular challenges 
of general $M$-estimators (versus $M(x)=|x|^p$, which we discuss more up to then). The technical overview
concludes with the formal statement of our regression result.

\textsf{Subspace Approximation} can be expressed in terms of a matrix
measure, defined as follows, that for some kinds of $M()$ is a norm.

\begin{definition}[definitions of $\norm{}_v$, $p$, $\norm{}_M$, $\norme{}$.]\label{def p vnorms}
For an $n \times d$ matrix $A$, we define
the \emph{$v$-norm} of $A$, denoted $\norm{A}_v$, to be
$\left[\sum_{1\le i\le n} M(\|A_{i,*}\|_2)\right]^{1/p}$, 
where $A_{i,*}$ is the $i$-th row of $A$, and
$p$ is a parameter associated with the function $M()$, which defines
a nice $M$-Estimator. We use the terminology $v$-norm, 
where $v$ stands for "vertical", to indicate that we take the sum of 
distances of the rows of the matrix. 
Here $M$ and $p$ will be understood from context;
our constructions never consider multiple $M$ and $p$ at the same time.
That is, for $M(x)=|x|^{p'}$,
the associated parameter $p$ is $p'$.
For a column vector $x$, we will write $\norm{x}_M=[\sum_{1\le i\le n} M(x_i)]^{1/p}$ for $\norm{x}_v$.
We also use an ``element-wise'' norm, with $\norme{A}^p$ equal to $\sum_{i\in[n], j\in[d]} M(A_{ij})$.
\end{definition}

This is the \emph{unweighted} $v$-norm; we will later use a version with weights. The ``$v$'' refers to
the ``vertical'' application of the $\ell_p$ norm. 

The subadditivity assumption for nice $M()$ implies that
$\norm{A-\hat A}_v$ is a metric on $A,\hat A \in\R^{n\times d}$, so that in particular
it satisfies the triangle inequality.
Using the polynomial upper bound and linear lower bound for $M()$, we have
for $\kappa \ge 1$,
\begin{equation}\label{eq scale insens}
(C_M\kappa)^{1/p} \norm{A}_v \le \norm{\kappa A}_v \le \kappa \norm{A}_v.
\end{equation}
While matrix norms satisfy the \emph{scale-invariance} condition
$\norm{\alpha A} = \alpha\norm{A}$ for all $\alpha\ge 0$,
here we will generally assume only this weaker condition of ``scale insensitivity.''
Despite this weaker condition, many constructions on metrics carry over,
as discussed in \S\ref{subsec scale insens}.

\subsubsection{Dimensionality reduction, I}


A prior result for $p = 2$ is that for suitable
$R\in\R^{d\times O(k/\eps)}$ randomly chosen so that the columns of $AR$ comprise $O(k/\eps)$
random linear combinations of the columns of $A$, it holds that
\[
\min_{\rank X = k}\|ARX-A\|_F \leq (1+\eps) \|A-A_k\|_F,
\]
where $A_k$ is the best rank-$k$ approximation to $A$ in Frobenius norm. 
The proof of this uses specific properties of the Frobenius norm such as approximate
matrix product \cite{s06,CW09,kn14} and the matrix Pythagorean theorem, and a 
natural question is if the same is true for any $p$. 

One of our key structural results is the following theorem, holding for nice $M$-estimators.

\begin{theorem}[A version of Theorem~\ref{thm good colspace sparse}]
\label{thm good colspace sparse informal}
If $R\in\R^{d\times m}$ is a sparse embedding matrix
with sparsity parameter $s$,
there is $s=O(p^3/\eps)$ and $m=O(k^2/\eps^{O(p)})=\poly(k/\eps)$ such that
with constant probability,
\begin{equation}\label{eq good colspace sparse informal}
\min_{\rank X = k} \norm{ARX-A}_v^p \le (1+\eps)\norm{A - A_k}_v^p.
\end{equation}
for $X$ of appropriate dimension. Here  $A_k \equiv \min_{\rank Y = k} \norm{Y - A}_v$.
\end{theorem}

Here the given matrix $R$ is a particular construction of
an $(\eps, \delta)$-subspace 
embedding for $k$-dimensional spaces.
A matrix is such an embedding if, for the row space of any fixed matrix $B$ of rank at most $k$, 
with probability at least $1-\delta$ it
holds that 
$\|yR\|_2 = (1 \pm \eps)\|y\|_2$ simultaneously for all $y$ in the row span of $B$.
To prove the above theorem, we show that 
if $R$ is an $(\eps^{p+1}, \eps^{p+1})$-subspace embedding for $k$-dimensional spaces, 
and $\|BR\|_v \leq (1+\eps^{p+1})\|B\|_v$ for a fixed matrix $B$, then
the theorem conclusion holds.
Using known 
subspace embeddings \cite{CW13,mm13,nn13,bdn15}, 
we can choose $R$ with $\poly(k/\eps)$ columns to satisfy \eqref{eq good colspace sparse informal} and moreover,
compute $AR$ in $O(\nnz(A)/\eps)$ time. We will later apply such $R$ with constant $\eps$. 

\begin{remark}
The above theorem, like many here, does not require $p<2$ for $M(x)=|x|^p$, or all the properties
of nice estimators $\cM_2$. We may not state all results in their fullest generality in this respect,
but there will be bounds ``$O(p)$'' that are unnecessary for our present algorithmic results.
\end{remark}

\subsubsection{Point reduction, I}\label{subsubsec point red 1}

As well as reducing the number of columns of the input matrix, we also need to reduce the number of rows.
Since $\norm{A}_v$ is based on the Euclidean norm of the rows $a_i$ (we will also write $A_{i*}$ for those rows),
many standard subspace embedding techniques can be applied ``on the right,'' taking $a_i$ to a smaller row $a_iR$,
with $\norm{a_i}_2\approx\norm{a_iR}_2$.
There are many fewer techniques applicable in our setting for application ``on the left,''
reducing the number of rows;
our algorithms perform all such reductions by sampling the rows.

A \emph{sampling} matrix $S$ is one whose rows are multiples of the
natural basis vectors $e_i, i\in[n]$. The sketch $SA$ has rows that are each multiples of some row
of $A$. Such sampling matrices will be found here
based on a vector $q\in\R^n$ of probabilities (with $q_i\in [0,1]$),
so that for each $i\in [n]$, the natural basis vector $e_i$
is independently chosen to be a row
of $S$ with probability $q_i$.  This implies that the number $m$ of rows
of $S$ is a random variable with expectation $\sum_i q_i$ (although indeed,
it is well-concentrated). We scale $e_i$ by $1/q_i^{1/p}$, for $M(x)=|x|^p$;
more generally we use a weighted version of the $v$-norm, since we cannot assume scale invariance.
With that scaling, $\E[\norm{S\hat A}_v] = \norm{S\hat A}_v$, for any $\hat A$.
(That is, any $\hat A$ that has $n$ rows; in general we assume that matrix operands
are conformable in shape for the operations done.)

The vector $q$ used for this importance sampling is based on
norms of rows of associated matrices; for example,
for the thin matrix $AR$ above, and $p=1$, we compute
a \emph{well-conditioned basis} for the columns of $AR$,
and $q_i$ is proportional to the $\ell_1$ norm of row $i$ of that basis.
Using these \emph{$\ell_1$-leverage scores} for sampling rows
goes back
to at least \cite{c05,dv07}. Algorithm~\ref{alg:M:informal}, using
such sampling, is a version of one of our algorithms. We note that it may be 
possible to further optimize the $\poly(k/\eps)$ factors in our algorithm using \cite{cp15}. 

A disadvantage of our sampling methods is that the sample
size depends on the number of columns of the matrix, so the row sample size for $AR$
can be much smaller than it would be for $A$; this is one reason that reducing the number of
columns is useful.

\begin{figure}
\begin{algorithm}[H]
\caption{$\textsc{ConstApprox}\cL_p(A, k)$}
\label{alg:M:informal}
(Simplified version of Algorithm~\ref{alg:M}, specialized to $\cL_p$)\\
{\bf Input:} $A \in \R^{n\times d}$, integer $k \geq 1$\\
{\bf Output:} $\hat X = UU^\top$, where $U\in\R^{d\times P_M}$ with orthonormal columns,
for a parameter $P_M$.
\begin{enumerate}
\item For parameter $m=\poly(k)$, let $R\in\R^{d\times m}$ be a sparse embedding matrix from
Theorem~\ref{thm good colspace sparse informal} with constant $\eps$
\item Compute a well conditioned basis of $AR$ (Def. \ref{def:wcb}, Thm.~\ref{thm well cond}),
and leverage scores $q'_i$ 
\item Let $S$ be a sampling matrix for $AR$,
	using probabilities $q_i \gets \min\{1, \poly(k)q'_i/\sum_i q'_i\}$
\item {\bf return} $\hat X = UU^\top$, where $U^\top$ is an orthonormal basis for the rowspace of $SA$.
\end{enumerate}
\end{algorithm}
\end{figure}

The next lemma is one we use for our analysis of this algorithm.
It claims a property for the sampled matrix that is cruder
than a subspace embedding, but holds for nice $M$-estimators,
and is used in our proof that Algorithm~\ref{alg:M:informal}
gives a bicriteria constant-factor approximate solution
for $M$-estimators with $M(x)=|x|^p$.

\begin{lemma}[A version of Lemma~\ref{lem S bb}]
\label{lem S bb informal}
Let $\rho>0$ and $B\in\R^{n\times r}$, with $r=\poly(k)$.
For sampling matrix $S$,
suppose for given $y\in\R^d$, with failure probability $\delta$ it holds
that $\norm{SBy}_M = (1\pm 1/10) \norm{By}_M$.
There is $K_1 = \poly(k)$ so that
with failure probability $\delta \exp(\poly(k))$,
any rank-$O(k)$ matrix $X\in\R^{d\times d}$
has the 
property that if $\norm{BX}_v \ge  K_1\rho$, then $\norm{SBX}_v \ge \rho$,
and that if $\norm{BX}_v \le \rho/K_1$, then $\norm{SBX}_v\le \rho$.
\end{lemma}

Our proof is roughly as follows.
We apply this lemma with $B=AR$. Letting $X_1$ be the minimizer
of $\norm{ARX-A}_v$ over rank-$k$ matrices, we use
the triangle inequality, so that for any $Y\in\R^{r\times d}$,
\[
\norm{S(ARY-A)}_v \ge \norm{S(ARY - ARX_1)}_v -  \norm{S(ARX_1 - A)}_v.
\]
We apply the lemma with $\rho = 10 \Delta_1$, letting $\Delta_1\equiv\norm{ARX_1-A}_v$,
to show that if $\norm{ARY - ARX_1}_v>K_1 10\Delta_1$, then $\norm{S(ARY - ARX_1)}_v\ge 10\Delta_1$.
Since $\E[ \norm{S(ARX_1 - A)}_v] =\Delta_1$,
with probability at least $4/5$, $\norm{S(ARX_1 - A)}_v\le 5\Delta_1$, so assuming that this and the
inequality from the lemma hold, we have $\norm{S(ARY - A)}_v \ge (10-5)\Delta_1$.
So any $Y$ with high cost $\norm{ARY - A}_v$ will have high estimated cost $\norm{S(ARY - A)}_v$,
and $X_2$ cannot be $Y$. The fact that $\Delta_1$ is not much larger than $\Delta^*$
implies that the matrix $X_2$ minimizing $\norm{SARX-SA}_v$ will have $\norm{ARX_2-A}_v$ within
a $\poly(k)$ factor of $\Delta^*$. Moreover, it is not hard to show that
the rows of $X_2$ are in the row space of $SA$, and therefore the projection $AUU^\top$ of
$A$ onto the row space of $SA$ has $\norm{A-AUU^\top}_v$ within a $\poly(k)$ factor of $\Delta^*$,
and the row space of $SA$ is a bicriteria $\poly(k)$-factor approximation.

(We may sometimes informally refer to $\poly(k)$ or $\poly(k/\eps)$ as ``constant,'' since our focus 
is removing dependence on $n$ and $d$.)

\subsubsection{Dimensionality reduction, II}

A  $\poly(k)$-factor bicriteria approximation $\hat X$ is useful in its own right,
but it can be used to obtain a different dimensionality reduction: a subspace, expressed as
the row space $F=\rspace(U^\top)$ for $U\in\R^{d\times\poly(k/\eps)}$ with orthonormal columns,
such that the optimum $k$-dimensional space contained in $F$ is
an $\eps$-approximate solution to the original problem, that is,
\begin{equation}\label{eq dv reduce}
\argmin_{\rank X=k}\norm{A - AUXU^\top}_v
\end{equation}
is an $\eps$-approximate solution for \textsf{Subspace Approximation}.

As noted above, the existence of a subspace of dimension
$\poly(k/\eps)$ that contains an approximate solution
was shown by Deshpande et al. \cite{dv07}. Here we extend their result in a few ways.
For one, we show that the claim holds for nice $M$-estimators as well as $M(x)=|x|^p$.

Another of our extensions is computational. The proof of \cite{dv07} is by way of an algorithm
that samples rows according to their residual distance to a subspace $V$, which
is initially $F$, and is extended by replacing $V$ by its span with each sampled row as it is chosen.
Such \emph{adaptive} sampling makes it impossible to achieve a running
time of $O(\nnz(A))$. We show that the same algorithm as in \cite{dv07} works even if the sampling is done
non-adaptively, that is, using distance to $F$. (Their proof also nearly applies.)
This may be of independent interest. Indeed, while for the Frobenius norm one can 
non-adaptively sample
with respect to the residual of a $\poly(k)$ approximation to refine to a $(1+\eps)$-approximation
\cite{dv06}, such
a result was not known for other loss functions $M$. Our formal statement, for a procedure \textsc{DimReduce},
is as follows. This procedure incorporates a scheme for fast estimation of residual norms (another of our extensions),
discussed in \S\ref{subsubsec lev score intro} below.

\begin{theorem}[A version of Theorem~\ref{thm non adapt}]
\label{thm non adapt informal}
Let $K>0$ and $\hat X\in\R^{d\times d}$ be a projection matrix such that $\norm{A(I-\hX)}_v \le K\Delta^*$;
as usual $\Delta^*\equiv \norm{A(I-X^*)}_v$, with $X^*\equiv \argmin_{\rank X=k}\norm{A(I-X)}_v$.
Then with small constant failure probability,
$\textsc{DimReduce}(A,k,\hX)$ returns $U\in\R^{d\times K \poly(k/\eps)}$ such that
\[
\min_{\rank X=k}\norm{A(I-UXU^\top)}_v\le (1+\eps)\Delta^*.
\]
The running time is $O(\nnz(A) + dK^2\poly(k/\eps))$  for $M(x)=|x|^p$  and
$O(\nnz(A)\log n+ dK^2\poly(k/\eps))$ for nice $M$-estimators.
\end{theorem}

\subsubsection{Point reduction, II}

The formulation \eqref{eq dv reduce} is computationally useful: for one, it allows use of sparse subspace embedding
matrices, so that there is a randomized construction
of $S\in\R^{\poly(k/\eps)\times d}$ such that with constant failure probability,
$\argmin_{\rank X=k}\norm{AS^\top - AUXU^\top S^\top}_v$ is an $\eps$-approximate
solution to \eqref{eq dv reduce}, and therefore to \textsf{Subspace Approximation}. That is,
by applying $S^\top$ to $A$ in $\nnz(A)$ time, and to $U$ in $d K \poly(k/\eps)$ time,
we have \emph{almost} removed $d$ from the problem.

What remains that depends on $d$
is $AU$, which we cannot afford the time to explicitly compute (taking $\Omega(\nnz(A)\poly(k/\eps))$ with standard methods).
However, the fact that $AU$ and $AS^\top$ have $\poly(k/\eps)$ columns
implies that row sampling can be applied effectively (since again, the row sample size depends
on the number of columns). Our strategy is to use row sampling, via probabilities proportional
to the leverage scores
of the thin matrix $[AS^\top\ AU]$, but we need to compute those leverage scores carefully,
without computing $AU$ explicitly. Having obtained those sampling probabilities,
we obtain a sampling matrix $T$. We now seek an approximate solution to
$\min_{\rank X=k}\norm{TAS^\top - TAUXU^\top S^\top}_v$,
a problem for which (for $M(x)=|x|^p$)  the dimensions $TAS^\top$, $TAU$, and $U^\top S^\top$ are all
in $\poly(k/\eps)$, and we can afford to compute them. (We compute $TAU$, for example, as $(TA)U$.)

\subsubsection{Solving Small Problems}

Finally, we need to solve this small problem.
While
there exist fairly involved net arguments (see, e.g., Section 5 of \cite{FMSW}) for solving small instances
of \textsf{Subspace Approximation}, at least for $p = 1$,
we formulate the problem as a system of polynomial
inequalities and immediately find this subspace in time $\exp(\poly(k/\eps))$ by a black box use of an
algorithm of Basu, Pollack, and Roy \cite{bpr96}. Our result is the following.

\begin{theorem}[A version of Theorem~\ref{thm small approx}]
\label{thm small approx intro}
Assume $p=a/b$ for integer constants $a,b\ge 1$, and let $\eps\in (0,1)$, and integer $k\in[0,m]$.
Given $A\in\R^{m'\times m}$, $B\in\R^{m\times m''}$, and $C\in\R^{m'\times m''}$, with $m',m'' = \poly(m/\eps)$,
a rank-$k$ projection matrix $X$ can be found that minimizes $\norm{AXB-C}_v^p$ up a $(1+\eps)$-factor,
in $\exp(\poly(m/\eps))$ time.
\end{theorem}

\subsubsection{Fast leverage score estimation}\label{subsubsec lev score intro}
An important consideration in our algorithms is the leading order term $\nnz(A)$; 
some parts of the analysis could be simplified if this were instead replaced
with $\nnz(A) \log n$, and if $\poly(k/\epsilon)$ is larger than $\log n$,
this is already a substantial improvement over the previous $\nnz(A) \poly(k/\epsilon)$
time algorithms. One may set $k$ and $1/\epsilon$ to be large if one is interested in a bicriteria
solution or dimensionality reduction, after which various heuristics can be 
run \cite{ding2006r,fl11,BWSDY}. 

However, if one is going to run an $\exp(\poly(k/\epsilon))$ time
algorithm on the small problem to find a rank-$k$ space, then it is also interesting
to allow $\poly(k/\epsilon) \leq \log n$. In this case, 
we still improve over prior work by achieving an optimal $O(\nnz(A))$ time, rather than
just $O(\nnz(A) \log n)$. 
This causes some complications in the computation of leverage scores;
as discussed in \S\ref{subsubsec point red 1}, some of our sampling matrices
use sampling probabilities proportional to leverage scores,
which are norms of well-conditioned bases.

\begin{definition}[Well-conditioned basis for the $p$-norm]\label{def:wcb}
An $n \times d$ matrix $U$ is an $(\alpha, \beta,p)$-well conditioned
basis for the column space of $A$ if, using $M(x)=|x|^p$,
(1) $\norme{U} \leq \alpha$ (where $\norme{}$ was defined in Def.~\ref{def p vnorms}), and
(2) for all $x \in \mathbb{R}^d$,
$\|x\|_{q} \leq \beta \|Ux\|_p$, where $1/p + 1/q = 1$. For ease of notation we will just say that
$U$ is a well-conditioned basis for $A$ if $\alpha, \beta = d^{O(p)}$, where $p$ is understood from 
context.  
\end{definition}

We use the following scheme to find well-conditioned bases.

\begin{theorem}\label{thm well cond}
Suppose $H\in\R^{d\times m}$. 
Suppose $\Pi\in\R^{s\times n}$ is an $\ell_p$ subspace embedding for the column space of $AH$, meaning
$\norm{\Pi AHx}_p^p = (1 \pm 1/2)\norm{AHx}_p^p$ for all $x \in \mathbb{R}^m$. Suppose we compute a $QR$-factorization
of $\Pi AH$, so that $\Pi AH = QR$, where $Q$ has orthonormal columns. Then $AHR^{-1}$ is a
$(\poly(m), 2, p)$-well conditioned basis for the column space of $AH$. 
There are $\ell_p$ subspace embeddings $\Pi$ with $s=\poly(m)$ for $p\in [1,2)$
that can be applied in $O(\nnz(A))$ time, so that $R^{-1}$ can be computed in
$O(\nnz(A)+\poly(m/\eps))$ time.
\end{theorem}

\ifCONF
\begin{IEEEproof}
\else
\begin{proof}
\fi
The existence of such $\Pi$ is shown  by \cite{mm13},
who also discuss the well-conditioned basis construction \cite{sw11}.
\ifCONF
\end{IEEEproof}
\else
\end{proof}
\fi

Given a well-conditioned basis $U$, here given implicitly as the product $AHR^{-1}$,
we need to estimate the norms of its rows. In prior work, this norm estimation was
done with a JL matrix, for example, a matrix $G\in\R^{m\times O(\log n)}$ of Gaussians
such that the row norms of $AHR^{-1}G$ are all approximately the same as those of
$AHR^{-1}$. We show that Gaussians with a constant number of rows, or even one row,
can be used, and still yield estimates that are algorithmically adequate.
We use a similar scheme for residual sampling.

\begin{theorem}[A special case of Theorem~\ref{thm G est}]
\label{thm G est informal}
Let $t_M\equiv 1$ for $M\in\cL_p$, and $t_M$ a large enough constant, for $M\in\cM_2$.
For matrix $U\in\R^{n\times d}$, suppose a sampling matrix $S$ using probabilities
$z_i\equiv\min \{1, r_1 z'_i/\sum_i z'_i\}$, where $z'_i=\norm{U_{i*}}_p^p$,
has small constant failure probability,
for some success criterion. (Here we require that the criterion allows oversampling.)
Let $G\in\R^{d\times t_M}$ be a random matrix with independent Gaussian entries
with mean 0 and variance $1/t_M$.
Then for $M\in\cL_p$,
a sampling matrix chosen with probabilities
\[
q_i\equiv\min \{1, K_2 d^{p/2} r_1^{p+1} q'_i/\sum_i q'_i\},
\]
where $q'_i \equiv |U_{i*}G|^p$,
also succeeds with small constant failure probability.
For $M\in\cM_2$, the same performance bound holds with $d^{p/2}r_1^{p+1}$
replaced by $r_1n^{O(1/t_M)}\log n$, with failure probability $1/n$.
\end{theorem}

Note that we apply the lemma to matrices $U$ with a small number of columns.

\subsubsection{Algorithms for $M$-Estimators}

Our results for general nice $M$-estimators, in $\cM_2$,
are weaker than for estimators in $\cL_p$.
There are various reasons for this, but the chief one is that effective row sampling matrices are
harder to come by. Leverage-score sampling is effective because of the bounds stated in the following
lemma. As applied to $\cL_p$, the stated bounds go back to  \cite{dv07}

\begin{lemma}[A version of Lemma~\ref{lem M lev}]
\label{lem M lev informal}
For nice $M$-estimators,
\[
\sup_{y\in\R^d} \frac{M(A_{i*}y)}{\norm{Ay}_M^p}
	\le \gamma_i(A,M) \equiv\max\{\beta\norm{U_{i*}}_p/C_M,\beta^p\norm{U_{i*}}_p^p\},
\]
where $U$ is an $(\alpha,\beta,p)$-well-conditioned basis of $A$.
For estimators in $\cL_p$, with $M(x)=|x|^p$ with $p\ge 1$,
$\gamma_i(A, M)$ can be sharpened to $\beta^p\norm{U_{i*}}_p^p$, and
for $\gamma(A,M)\equiv \sum_i \gamma_i(A,M)$ we have
$\gamma(A,M) = O((\alpha\beta)^p) = O(d^{\max\{p,1+p/2\}})$.
For $M\in \cM_2$, a general nice $M$-estimator with $p\le 2$,
$U$ can be an orthogonal basis of $A$, and
$\gamma(A,M) \equiv \sum_i \gamma_i(A,M) = O(\sqrt{{d} n})/C_M$.
\end{lemma}

The quantity $\frac{M(A_{i*}y)}{\norm{Ay}_M^p}$ is a kind of \emph{sensitivity} score \cite{vx12},
capturing how much effect the $i$'th summand $M(A_{i*}y)$ can have on the sum $\norm{Ay}_M^p$
of all such values. The critical quantity is the total $\gamma(A,M)$ of these sensitivities, which
determines the row sample size. Where for $\cL_p$, that size is $\poly(d)$,
for $\cM_2$ it is $\Omega(\sqrt{n})$. That is, for $M$-estimators, row sampling
only reduces the problem size from $n$ to $O(\sqrt{n})$ as a function of $n$, and recursive applications
of sampling are needed to get problems down to $\poly(d)$.

\subsubsection{Regression} 

A simple byproduct of our machinery for \textsf{Subspace Approximation} is
a regression algorithm for convex $M$-estimators.

\begin{theorem}\label{thm regress}
For vector $b\in\R^n$ and convex $M\in\cM_2$,
there is a procedure that with small constant failure probability
finds
an $\eps$-approximate solution to $\min_{x\in\R^d}\norm{Ax-b}_v$.
The procedure takes $O(\nnz(A) + \poly(d/\eps))$ time.
\end{theorem}

We use fast leverage score estimation, and row sampling. This broadens
the results of  \cite{cw15}, where a similar result was shown for the Huber estimator only.
\ifCONF
\else
The proof is in \S\ref{sec regress}.
\fi

\section{Notation and Terminology}\label{sec:prelim}

Again, throughout we assume that $M()$ and norm parameter $p$ from Def. \ref{def p vnorms} are fixed:
the constant factors in $O()$ may depend on $p$,
and various norms will implicitly depend on $p$. In $\poly()$, such as $\poly(k/\eps)$,
the degree may depend on $p$. As noted in the introduction, our main results are for $p \in [1,2)$,
in which case the $O(p)$ term is just $O(1)$.

In this paper, unless otherwise indicated,
$A$ is an $n\times d$ matrix,
matrix $B\in\R^{n\times d'}$, and
constraint set $\cC\subset\R^{d\times d'}$. Throughout we assume that the
error parameter $\eps$ is smaller than an appropriate constant.

Let $A_{i*}$ denote $a_i$, the $i$-th row of $A$, and $A_{*j}$ denote the $j$-th column.

\begin{definition}[weighted $\norm{\,}_v$, $\norme{\,}$, $\norm{\,}_M$]\label{def w8 v}
For $w\in\R^n$ with all $w_i\ge 1$,  and $M\colon \R\mapsto \R^+$ and $p\ge 1$ as is
Def. \ref{def p vnorms}, let
\[
\norm{A}_v\equiv \left[ \sum_{i\in [n]} w_i M(\norm{A_{i*}}_2)\right]^{1/p},
\]
and
\[
\norm{A}_h \equiv \norm{A^\top}_v = \left[ \sum_{j\in [d]} w_i M(\norm{A_{*j}}_2)\right]^{1/p},
\]
and let $\norme{A}$ denote $\left[\sum_{i,j} w_i M(A_{i,j}) \right]^{1/p}$.
For a given vector $x\in\R^m$, let
$\norm{x}_M \equiv \left[\sum_{i\in [m]} w_i M(x)\right]^{1/p}$.
The weight vector $w$ will be
generally be understood from context. When the relevant weight $w$ needs emphasis, we may write
$\norm{A}_{v,w}$ or $\norm{x}_{M,w}$.
\end{definition}

\begin{definition}[$X^*$,$\Delta^*$]\label{def X*}
Let
\[
X^*\equiv\argmin_{\rank X=k}\norm{A(I-X)}_v,
\]
with $\Delta^*\equiv \norm{A(I-X^*)}_v$. More generally, $X^*$ and $\Delta^*$ will be the optimum and its cost for the problem
under consideration.
\end{definition}

Note that $X^*$ will be a projection matrix
(otherwise $XY$, for $Y$ the projection onto the rowspan of $A$, would give
a better solution).


\begin{definition}[even, monotone, polynomial, linearly bounded, subadditive, nice]
As discussed above, we will need $M$ to be \emph{nice}, with these properties:
\begin{itemize}
\item \emph{even}, so $M(a) = M(-a)$;
\item \emph{monotone}, so that $M(a)\ge M(b)$ for $|a|\ge |b|$; and
\item \emph{polynomially bounded with degree $p$}, for some $p$, meaning that
$M(a)\le  M(b)(a/b)^p$ for all $a,b\in\R^+$ with $a\ge  b$. Since $p$ is fixed throughout,
we will just say that $M$ is polynomial.
\item \emph{linearly bounded below}, that is, there is some $C_M>0$ so
that $M(a)\ge C_M\frac{|a|}{|b|}M(b)$ for all $|a|\ge |b|$.
\item \emph{$p$-th root subadditive}, that is, $M(x)^{1/p}$ is subadditive, so that $\norm{A}_v$ satisfies the triangle inequality.
\end{itemize}
\end{definition}

The subadditivity assumption implies that
$\norm{A-\hat A}_v$ is a metric on $A,\hat A \in\R^{n\times d}$. We will also use,
for $X,Y\in\R^{d\times d}$, the ``norm'' $\norm{AX}_v$
and pseudometric $D_v (X,Y) \equiv \norm{A(X-Y)}_v$: the only property
of a metric that it lacks is ``identity of indiscernables'', since it may report the distance of
$X$ and $Y$ as zero when $X\ne Y$. 
Note that if $M(x)$ is subadditive,
so is $M(x)^{1/2}$.

It will be helpful that
\begin{equation}\label{eq M add}
M(a+b)\le M(2\max\{a,b\}) \le 2^p M(\max\{a,b\}) \le 2^p(M(a) + M(b)),
\end{equation}
using monotonicity and the polynomial bound.

\begin{definition}[$A_k$, $A^+$]
Let
\[
A_k \equiv \argmin_{\rank Y = k}\norm{Y- A}_v = \argmin_{\rank Y = k} \norm{Y^\top - A^\top}_{h}.
\]
Let $A^+$ denote the Moore-Penrose pseudo-inverse of $A$,
$A^+ = U\Sigma^{-1}V^\top$, where $A=U\Sigma V^\top$ is the thin
singular value decomposition of $A$.
\end{definition}

\begin{definition}[$\eps$-contraction, $\eps$-dilation, $\eps$-embedding]
For a matrix measure $\norm{}$, and $\cT\subset\R^{n\times d'}$,
call matrix $S$ an $\eps$-contraction
for $\cT$ with respect to $\norm{}$
if $\norm{SY}\ge (1-\eps)\norm{Y}$ for all $Y\in\cT$. 

Similarly, call $S$ an $\eps$-dilation if $\norm{SY}\le (1+\eps)\norm{Y}$ for all $Y\in\cT$.

Say that $S$ is an $\eps$-embedding for $\cT$ with respect to $\norm{}$
if $S$ is both an $\eps$-contraction and $\eps$-dilation for
$\cT$ with respect to $\norm{}$.
\end{definition}

When $\cT$ is a singleton set $\{B\}$, we will refer to $B$ instead of $\{B\}$ when using these terms.


\begin{definition}[$\rspace$, $\cspace$]
The row space $\rspace(A)$ is defined as $\rspace(A) \equiv \{x^\top A\mid x\in\R^n\}$,
and similarly the column space $\cspace(A)\equiv \{Ax\mid x\in\R^d\}$.
\end{definition}

\begin{definition}[\emph{subspace} embedding, contraction, dilation]
When
$S$ is an $\eps$-embedding for $\cspace(A)$ with respect to
$\norm{}_2$, say that $S$ is a \emph{subspace} $\eps$-embedding
for $A$; that is, $\norm{SAx} = (1\pm\eps)\norm{Ax}$ for all $x\in\R^d$.
Similarly define subspace $\eps$-contraction and $\eps$-dilation.
\end{definition}

\begin{definition}[\emph{affine} embedding, contration, dilation]
When $S$ is an $\eps$-embedding for $\{AX - B\mid X\in\R^{d\times d'}\}$ with respect
to some $\norm{}$, say that $S$ is an \emph{affine}
$\eps$-embedding for $(A,B)$ with respect to $\norm{}$.
Similarly define affine $\eps$-contraction and $\eps$-dilation.
\end{definition}

\begin{definition}[\emph{lopsided} embeddings]
When $S$ satisfies the following conditions for some constraint set $\cC$
and norm $\norm{}$ (or even, any nonnegative function),
say that $S$ is an \emph{lopsided $\eps$-embedding} for $(A,B)$ with respect to
$\cC$ and $\norm{}$:
\begin{enumerate}[i.]
\item $S$ is an affine $\eps$-contraction for $(A,B)$, and
\item $S$ is an $\epsilon$-dilation for $B^*$, where $B^*\equiv AX^*-B$, and
	$X^*=\argmin_{X\in\cC}\norm{AX-B}$.
\end{enumerate}
\end{definition}


\section{Sparse affine lopsided embeddings}\label{sec sparse_embed}

The following lemma is key to our results.

\Ken{all results in this section require $M$ only to be even, monotone, and having upper-bounded growth.
A weight vector $w$ is assumed throughout: the notation implicitly assumes weights.}

\begin{lemma}\label{lem contraction p}
Suppose $M\colon \R\mapsto \R^+$ is even, monotone and
polynomial.
Let $X^*\equiv \argmin_X \norm{AX-B}_h$,
$B^* \equiv AX^* - B$,  and  $\Delta^*\equiv\norm{B^*}_h^p$.
Let $S\in\R^{m\times n}$ be a random matrix
with the properties that:
\begin{enumerate}[i.]
\item $S$ is a subspace $\epsilon$-contraction for $A$ with respect to $\norm{}_2$;
\item for all $i \in [d']$, $S$ is a subspace $\eps^{p+1}$-contraction for $[A\ B_{*i}]$
with respect to $\norm{}_2$,
with probability at least $1-\epsilon^{p+1}$;
\item $S$ is an $\eps^{p+1}$-dilation for $B^*$ with respect to $\norm{}_h$, that is,
$\norm{SB^*}_h\le (1+\eps^{p+1})\norm{B^*}_h$.
\end{enumerate}
Fix $\delta\in (0,1)$.
With failure probability at most $\delta$,
$S$ is an affine $O(\eps)$-contraction
for $(A,B)$ with respect to $\norm{}_h$,
meaning that for all $X\in\R^{d\times d'}$,
it holds that 
\[\norm{S(AX-B)}_h \ge (1- O(\epsilon))\norm{AX-B}_h.
\]
\end{lemma}

The lemma is very general, but in fact holds for an even broader class of matrix measures,
where the Euclidean norm appearing in the definition of $\norm{}_h$ is generalized to be
an $\ell_q$ norm.

\ifCONF\begin{IEEEproof}
\else
\begin{proof}
\fi
Let $\delta_i = \norm{B^*_{*i}}_2$
and $h_i = \norm{SB^*_{*i}}_2$.
For $i\in[d']$, let $Z_i$ be an indicator random variable where
$Z_i=0$ if $S$ is an $\epsilon^{p+1}$-contraction for $B_{*i}$,
and $Z_i=1$ otherwise. If $Z_i=1$ call $i$ \emph{bad}, otherwise $i$ is \emph{good}.

Consider arbitrary $X\in\R^{d\times d'}$.

Say a bad $i$ is \emph{large} if
$\norm{(AX-B)_{*i}}_2 \geq \eps^{-1}(\delta_i + h_i)$;
otherwise a bad $i$ is \emph{small}. 
Then using \eqref{eq M add} and the polynomial bounded condition on $M$, 
\begin{equation}\label{eqn:bound B*}
\sum_{\text{small\ }i} w_i M(\norm{(AX-B)_{*i}}_2)
	   \le \sum_{\text{bad\ }i}   \eps^{-p} w_iM(\delta_i + h_i) 
	 \le  \eps^{-p} 2^p \sum_{\text{bad\ }i} w_i (M(\delta_i) + M(h_i)).
\end{equation}

Using (ii), $\E[\sum_{\text{bad\ }i} w_i M(\delta_i)] \le \eps^{p+1}\Delta^*$,
so by a Markov bound
\begin{equation}\label{eq delta bound}
\sum_{\text{bad\ }i} w_i M(\delta_i) \le C\eps^{p+1} \Delta^*,
\end{equation}
for constant $C$ with failure probability at most $1/C$. Assume the event that \eqref{eq delta bound} holds.

Similarly,
\begin{align*}
\sum_{\text{bad\ } i} w_i M(h_i) 
&  = \norm{SB^*}_h^p  - \sum_{\text{good\ }i} w_i M( \norm{SB^*_{*i}}_2) \\
 & \le  (1+\eps^{p+1}) \Delta^* -  (1-\eps^{p+1})^p \sum_{\text{good\ }i} w_i M(\norm{B^*_{*i}}_2) & \text{by (iii), (ii), polynomial bounded condition} \\
& \le (1+\eps^{p+1}) \Delta^* - (1-\eps^{p+1})^p(1-C\eps^{p+1}) \Delta^* & \text{by \eqref{eq delta bound}}\\
& = O(\eps^{p+1} \Delta^*)
\end{align*}

Returning to (\ref{eqn:bound B*}), we have 
\begin{align}\label{eqn:bound2 small i}
\sum_{\text{small\ }i} w_i M(\norm{(AX-B)_{*i}}_2)
	   & \le \eps^{-p}2^p\left[
			\sum_{\text{bad\ } i} w_i M(\delta_i) + \sum_{\text{bad\ } i} w_i M(h_i)\right] \nonumber
	\\ & \le \eps^{-p}2^p (C\eps^{p+1}\Delta^* + O(\eps^{p+1}\Delta^*))\nonumber
	\\ & = O(\eps \Delta^*).
\end{align}

For arbitrary $X$ we have
\begin{align}\label{eq bound on bad}
\sum_{\text{bad\ } i} w_i M(\norm{S(AX-B)_{*i}}_2) \nonumber
	   & \ge \sum_{\text{large\ } i} w_i M(\norm{S(AX-B)_{*i}}_2) \nonumber
	\\ & \ge \sum_{\text{large\ } i} w_i M(\|S(AX^*-AX)_{*i}\|_2 - h_i) \nonumber
	\\ & \ge \sum_{\text{large\ } i}  w_i M((1-\eps)\|(AX^*-AX)_{*i}\|_2 - h_i) \nonumber
	\\ & \ge \sum_{\text{large\ } i} w_i  M((1-\eps)\|(AX-B)_{*i}\|_2 - \delta_i - h_i) \nonumber
	\\ & \ge (1-O(\eps))\sum_{\text{large\ } i} w_i M(\norm{(AX-B)_{*i}}_2),
\end{align}
where  the first inequality uses that all large $i$ are bad by definition,
the second inequality
is the triangle inequality, the third inequality is that $S$ is a subspace $\eps$-contraction
for $A$,  the fourth inequality is the triangle inequality, and the last inequality uses
the definition of large and the polynomial growth bound for $M()$.

It follows that
\begin{align*}
\|S(AX-B)\|_h^p
& =  \sum_{\text{good\ }i} w_i M(\norm{S(AX-B)_{*i}}_2)
	+ \sum_{\text{bad\ } i} w_i M(\norm{S(AX-B)_{*i}}_2)\\
& \ge (1-\eps)^p \sum_{\text{good\ }i} w_i M(\norm{(AX-B)_{*i}}_2)
	+ \sum_{\text{bad\ }i} w_i M(\norm{S(AX-B)_{*i}}_2) \\
& \ge  (1-\eps)^p \sum_{\text{good\ }i} w_i M(\norm{(AX-B)_{*i}}_2)
	+ (1-O(\eps))\sum_{\text{large\ } i} w_i M(\norm{(AX-B)_{*i}}_2)\\
& \ge (1-O(\eps))\norm{AX-B}_h^p - O(\eps\Delta^*) \\
& \ge (1-O(\eps))\norm{AX-B}_h^p
\end{align*}
where the first inequality uses that $S$ is a subspace embedding $[A\ B_{*i}]$ for good $i$, 
the second inequality uses \eqref{eq bound on bad},
and third inequality uses (\ref{eqn:bound2 small i}), and the last
uses that $\Delta^*\le \norm{AX-B}_h^p$ by definition.
The lemma follows.
\ifCONF
\end{IEEEproof}
\else
\end{proof}
\fi

\begin{lemma}\label{lem hold}
For random variable $X$, $p\ge 1$, and $\alpha>0$,
$\E[|X-1|^p] \le \alpha^p$ implies
$\E[|\max\{|X|^p,1\}] \le (1+\alpha)^p$.
\end{lemma}

\ifCONF\else
\begin{proof}
We have
\[
\E\max\{[|X|^p,1\}] 
	\le \E[(1 + |X-1|)^p]
	= \E[\sum_i \binom{p}{i} |X-1|^i]
	\le \sum_i \binom{p}{i} (\alpha^p)^{i/p}
	= (1+\alpha)^p,
\]
where the second inequality follows from H\"older's inequality.
\end{proof}
\fi 

\begin{lemma}\label{lem lopsided sparse}
If $S\in R^{m\times d}$ is a sparse embedding matrix\cite{bn13,CW13,mm13,nn13}, then there is
$m=O(d^2/\eps^{O(p)})=\poly(d/\eps)$ such that $S$ is a lopsided $\eps$-embedding
for $(A,B)$ with constant probability for $\norm{}_h$,
for $M$ even, monotone, and polynomial.
The product $SA$ is computable in $O(s\nnz(A))$ time,
where $s=O(p^3/\eps)$. A value of $s=O(p^3)$ can be used if $m$ is increased
by an additive $\eps^{-O(p^2)}$.
\end{lemma}

\ifCONF\else
\begin{proof}
We show that the conditions of Lemma~\ref{lem contraction p} hold
for the OSNAP construction of \cite{nn13}, with sparsity parameter $s$.

\Ken{OK now?}
A sparse embedding matrix $S$ of either given dimensions satisfies (i) and (ii)
of Lemma~\ref{lem contraction p} by Theorem~3 of \cite{nn13},
taking $\delta$ and $\eps$ of that lemma to be $\eps^{p+1}$.
That implies that $m=O(d^2/\eps^{O(p)})$ suffices for $s=1$. Moreover,
increasing $s$ does not degrade the quality bounds.


Next we show that $S$ satisfies (iii) of Lemma~\ref{lem contraction p} (and so also condition
(ii) in the definition of lopsided embeddings).
Let $y\in\R^n$ be a unit vector.
From (20) of \cite{nn13}, within the proof of Theorem 9, and with $d$ of that theorem equal to 1, we have for even $\ell\ge 1$,
\[
\E[(\norm{Sy}^2 - 1)^\ell] \le e\ell^{4.5}\max_{2\le b\le \ell} (b^3/s)^{\ell-b} (b^4/e)^b m^{-b/2}.
\]
If $s\ge e^2\ell^3/\eps^{p+1}$ and $m\ge e^3\ell^8/\eps^{2(p+1)}$,
then the upper bound is
\[
\eps^{\ell(p+1)} e\ell^{4.5} \max_{2\le b\le\ell} (b/\ell)^{3\ell+b} \exp(-2\ell-b/2),
\]
and this is at most $\eps^{\ell(p+1)}$, since $b/\ell\le 1$ and
$e\ell^{4.5}\exp(-2\ell-b/2) \le \ell^{4.5}\exp(-2\ell) < 1$ for $b\ge 2$
and $\ell\ge 1$.
Similarly, $s\ge e^2\ell^3$ and $m\ge e^3\ell^8/\eps^{\ell (p+1)}$
yield the same bound: the part depending on $\eps$ is shifted to $m$.
In either case,
$\E[(\norm{Sy}^2 - 1)^\ell] \le \eps^{\ell(p+1)}$.
Using Lemma~\ref{lem hold} with even $\ell=p/2$ and unit vector $y$,
$\E[\max\{\norm{Sy}^p,1\}] \le (1+\eps^{p+1})^{p/2}\le 1+p\eps^{p+1}$,
and so for $y$ not necessarily
unit, and letting $a_+\equiv \max\{a,0\}$ for $a\in\R$,
\[
\E[(\norm{Sy}^p/\norm{y}^p - 1)_+ ] \le p \eps^{p+1}.
\]
(If $p$ is not divisible by $4$, we apply the lemma to $p'\le p+3$,
and pay a constant factor in the sizes of $s$ and $m$.)
Applying the bound to each column of $B^*$, 
\begin{align*}
\E[(\norm{SB^*}_{h}^p - \norm{B^*}_h^p)_+]
	   & \le \sum_i w_i \E [ (M(\norm{SB^*_{*i}}) - M(\norm{B^*_{*i}}))_+]
	\\ & \le \sum_i w_i M(\norm{B^*_{*i}})
			\E[(\frac{\norm{SB^*_{*i}}^p}{\norm{B^*_{*i}}^p} - 1)_+]  & \mathrm{using\ poly\ growth}
	\\ & \le \sum_i w_i M(\norm{B^*_{*i}}) p\eps^{p+1}
	\\ & = \norm{B^*}_p^p p\eps^{p+1}.
\end{align*}
From Markov's inequality, we have
$(\norm{SB^*}_{h}^p - \norm{B^*}_{h}^p)_+ \le 10p \eps^{p+1} \norm{B^*}_{h}^p$ with 
failure probability $1/10$. Applying this to $\eps' = (10p)^{-1/(p+1)}\eps$,
condition (iii)
of Lemma~\ref{lem contraction p} holds, and so condition (ii) defining
lopsided embeddings.
Thus the claim holds for sparse embedding $S$ of the given size and sparsity,
and the lemma follows.
\end{proof}
\fi 

\Ken{
Note that the bound on $m$ needed for (i) and (ii), with a factor of $d^2$, implies
that the bound on $\norm{Sy}$ in the proof for (iii) may be quite conservative, since
the factor of $d^2$ implies a corresponding reduction in the bound for
$\E[(\norm{Sy}^2-1)^\ell]$.
}

\begin{lemma}\label{lem lopsided}
For a nonnegative matrix function $\norm{}$,
suppose that $S$
is a lopsided $\eps$-embedding for $(A,B)$.
Then if $\tilde X\in\R^{d\times d'}$ has
$
\norm{S(A\tilde X - B)} \le \kappa \min_{X\in\cC}\norm{S(AX-B)}$
for some $\kappa$, then
$\norm{A\tilde X-B}
	\le \kappa (1+3\eps)\Delta^*.$
\end{lemma}

\ifCONF\else
\begin{proof}
Using the hypotheses,
\begin{align*}
\norm{A\tilde X - B}
	   & \le \norm{S(A\tilde X - B)}/(1-\eps) & \text{\ by\ (i) of lopsided}
	\\ & \le \kappa\norm{S\hat B}/(1-\eps) &\text{\ def\ of\ }\tilde X
	\\ & \le \kappa (1+\eps)\norm{B^*}/(1-\eps) & \text{\ by (ii) of lopsided}
	\\ & \le \kappa(1+3\eps)\norm{AX^*-B}.
\end{align*}

The lemma follows.
\end{proof}
\fi 

\begin{lemma}\label{thm good colspace}
If $R\in\R^{d\times m}$ has that
$R^\top$ is a lopsided $\eps$-embedding for $(A_k^\top, A^\top)$ with respect
to $\norm{}_h$, 
then
\begin{equation}\label{eq good colspace}
\min_{\rank X = k} \norm{ARX-A}_v^p \le (1+3\eps)\norm{A_k- A}_v^p
\end{equation}
for $X$ of the appropriate dimensions. 
\end{lemma}

\ifCONF\else
\begin{proof}
Apply Lemma~\ref{lem lopsided},
for $A_k^\top\in\R^{d\times n}$ taking the role
of $A$ in the lemma,
$A^\top$ the role of $B$,
$R^\top\in\R^{m\times d}$ the role of $S$,
$\cC=\R^{n\times n}$ that of $\cC$,
$I$ the role of $X^*=\argmin_X\norm{A_k^\top X -A^\top}_{h}$,
and $A_k^\top - A^\top$ the role of $B^*$.

Lemma~\ref{lem lopsided} implies that for
$\tilde Y \equiv \argmin_Y\norm{R^\top (A_k^\top Y - A^\top)}_{h}$, we have
\[
\norm{A_k^\top\tilde Y - A^\top}_{h}^p
	\le (1+3\eps)\norm{A_k^\top - A^\top}_{h}^p
\]
Noting that here $\tilde Y = (R^\top A_k^\top)^+ R^\top A^\top$, by
taking the transpose we have
\[
\norm{AR((R^\top A_k^\top)^+)^\top A_k - A}_v^p
	\le (1+3\eps)\norm{A_k - A}_v^p
\]
and since by definition $\min_{\rank X = k} \norm{ARX-A}_v^p$ is no
more than the left hand side of this
inequality,  the lemma follows.
\end{proof}
\fi 

\begin{theorem}\label{thm good colspace sparse}
If $R\in\R^{d\times m}$ is a sparse embedding matrix
with sparsity parameter $s$,
there is $s=O(p^3/\eps)$ and $m=O(k^2/\eps^{O(p)})=\poly(k/\eps)$ such that
with constant probability,
\begin{equation}\label{eq good colspace sparse}
\min_{\rank X = k} \norm{ARX-A}_v^p \le (1+3\eps)\norm{A_k- A}_v^p
\end{equation}
for $X$ of the appropriate dimensions. 
\end{theorem}

\ifCONF\else
\begin{proof}
Note that $R^\top$ is a lopsided embedding for $(A_k^\top, A^\top)$ if and only if
it is a lopsided embedding for $(V_k, A^\top)$, where $V_k$ comprises
a basis for the columnspace of $A_k^\top$.
Lemma~\ref{lem lopsided sparse} implies that with the given bound
on $m$, $R^\top$ is a lopsided embedding for $(V_k, A^\top$),
and so for $(A_k^\top, A^\top)$. This condition and Lemma~\ref{thm good colspace}
implies the lemma.
\end{proof}
\fi


\section{Sampling matrices for low-rank approximation}\label{sec sampling}

\subsection{Nets, Bounds, Approximations for Scale-Insensitive Measures}\label{subsec scale insens}

As noted in the introduction, most of the proposed $M$-estimators yield
measures on matrices that are ``almost'' norms; the main property they lack
is scale invariance. However, most proposed $M$-estimators do satisfy
the weaker ``scale insensitivity'' of \eqref{eq scale insens}.
In this subsection, we give some lemmas regarding such scale-insensitive
almost-norms, that are weaker versions of properties held by norms.

In this subsection only, $\norm{}$ denotes a measure on a $d$-dimensional
vector space $\cal V$ such that $\norm{0} = 0$, $\norm{x}\ge 0$ for all $x\in\cal V$,
$\norm{x+y}\le \norm{x} + \norm{y}$, and $\norm{}$ satisfies
\begin{equation}\label{eq scale insens gen}
(C_M\kappa)^{1/p} \norm{x} \le \norm{\kappa x} \le \kappa \norm{x},
\end{equation}
for all $\kappa\ge 1$. This implies a continuity condition,
that for any $x\in\cV$ with $\norm{x}\ne 0$ and $\rho > 0$, there is $\beta>0$
so that $\norm{\beta x} = \rho$.

Also $\norm{}_S$ denotes a measure on $\cV$
satisfying the same conditions.

Let $\Vol$ be a \emph{nonnegative measure} on $\cV$, so that if $W,Z\subset \cV$ are disjoint, then
$\Vol(W\cup Z) = \Vol(W)+\Vol(Z)$, and if $W\subset Z$, then $\Vol(W)\le \Vol(Z)$, and
for $\alpha>0$, $\Vol(\alpha W) = \alpha^d\Vol(W)$.
	
Let $\cB_\rho$ denote the ball $\{x\in\cV\mid \norm{x}\le \rho\}$,
and let $\cS_\rho$ denote the sphere $\{x \in \cV\mid \norm{x}=\rho\}$.
Note that for any $y\in\cV$, the ball
$\{x\in\cV\mid \norm{x-y}\le \rho\} = \cB_\rho+y$,
and similarly for spheres.

An $\eps$-cover of $\cC\subset\cV$ is a collection $\cN\subset \cV$
such that for all $y\in\cV$ there is some $x\in\cN$ with $\norm{x-y}\le \eps$.

\begin{lemma}\label{lem v-eps-net gen}
Let $\cC\subset\cV$. For given $\eps>0$,
$\cB_\rho\cap\cC$ has an $\eps \rho$-cover in $\cC$ of size $(4^p/C_M\eps^p)^d$.
\end{lemma}

\ifCONF\else
\begin{proof}
First, assume that $\cC=\cV$.
If  $x\in \cB_\rho$ then
$\norm{\alpha x} \le \eps \rho$, where $\alpha\equiv C_M \eps^p$,
since from \eqref{eq scale insens gen} with $\alpha = 1/\kappa$,
\[
\norm{\alpha x}^p
	\le (\alpha/C_M) \norm{x}^p
	\le \eps^p \rho^p.
\]
Thus $\alpha \cB_\rho\subset \cB_{\eps \rho}$.

We have $\Vol(\cB_{\eps \rho})\ge \Vol(\alpha \cB_\rho) = \alpha^{d}\Vol(\cB_\rho)$.
Thus at most $\alpha^{-d}$ translations of $\cB_{\eps \rho}$ can be packed
into $\cB_\rho$ without overlapping. Let $\cN'$ be the collection of centers
of such a maximal packing. Then
every point of $\cB_\rho$ must be within $2\eps \rho$ of a point of $\cN'$,
since otherwise another translation of $\cB_{\eps \rho}$ would fit.
There is therefore  a $2 \eps \rho$-cover of $\cB_\rho$ of size at most
$(C_M \eps^p)^{-d}$, and so an $\eps \rho$-cover $\cN$ of size at
most $(2^p/C_M \eps^p)^{d}$.

Now to consider $\cC\subset\cV$. For each $x\in\cN$, let $\cC_x$ denote
the subset of vectors in $\cC\cap \cB_\rho$ with $x$ closest. Pick an arbitrary $x'\in\cC_x$,
so that all members of $\cC_x$ have $x'$ within $2\eps \rho$. The resulting collection
$\cN'$ is a $2\eps \rho$-cover of $\cC\cap \cB_\rho$.
Therefore there is an $\eps \rho$-cover of $\cC\cap \cB_\rho$
comprising members of $\cC$, of size at most $(4^p/C_M\eps^p)^{d}$.
The lemma follows.
\end{proof}
\fi 

\begin{lemma}\label{lem v-dilation gen}
\begin{enumerate}[i.]
\item
If $\norm{x}_S/\rho \le 1+\eps$
for all $x$ in a $C_M\eps^p \rho$-cover $\cN$ of $\cS_\rho$, then
$\norm{x}_S/\rho \le 1+3\eps$ for all $x\in\cS_\rho$.

\item If for $\gamma\ge\eps$,
\[
1-\eps \le \norm{x}_S/\rho \le 1+\gamma,
\]
for all $x\in\cN$ of $\cS_\rho$,
then
\[
\norm{x}_S/\rho \ge 1 - \eps(2+3\gamma)
\]
for all $x\in\cS_\rho$.
\end{enumerate}
\end{lemma}

\ifCONF\else
\begin{proof}
We adapt an argument from Lemma 9.2 of
\cite{Ball}.
Let $\cN$ be a $C_M\eps^p \rho$-cover of $\cS_\rho$ for which
$\norm{}_S$ satisfies the dilation condition of (i) in the lemma.
Let $\eta\equiv \sup_{x\in \cS_\rho} \norm{x}_S/\rho$,
realized by $x_s$. Then for $x_\delta$ where
$x_\delta\equiv x_s - x'$ with $x'\in\cN$
such that $\norm{x_\delta}\le C_M\eps^p \rho$,
pick $\kappa$ such that $\kappa x_\delta \in \cS_\rho$.
This implies $\kappa \ge \rho/\norm{x_\delta}\ge 1/C_M\eps^p$,
using \eqref{eq scale insens gen}.

We have
\[
\norm{x_\delta}_S
	\le (C_M\kappa)^{-1/p} \norm{\kappa x_\delta}_S
	\le \frac{(C_M\eps^p)^{1/p} }{C_M^{1/p}} \eta \rho
	\le  \eps\eta \rho,
\]
so that
\[
\eta \rho = \norm{x_s}_S
	\le \norm{x'}_S + \norm{x_s - x'}_S
	\le \rho (1+ \eps) + \eps\eta \rho
\]
and so $\eta \le (1+\eps)/(1 - \eps ) \le 1+ 3\eps$, implying
$\norm{x}_S/\rho \le 1+3\eps$ for all $x\in \cS_\rho$.
Claim (i)
follows.

For claim (ii), the conditions readily imply, similarly to claim (i),
that $\norm{x}_S \le (1+3\gamma) \rho$
for all $x\in\cS_\rho$. For given $x\in\cS_\rho$, pick $x'\in\cN$ such that for $x_\delta\equiv x - x'$,
$\norm{x_\delta}\le C_M \eps^p \rho $. Then similarly to the argument in claim (i),
$\norm{x_\delta}_S/\rho \le \eps(1+3\gamma)$, and so
\[
\norm{x}_S \ge \norm{x'}_S - \norm{x_\delta}_S
	\ge \rho(1-\eps) - \eps(1+3\gamma)\rho
	= \rho(1 -\ \eps(2+3\gamma)),
\]
and claim (ii), and the lemma, follow.
\end{proof}
\fi 

\begin{lemma}\label{lem weak embed gen}
\begin{enumerate}[i.]
\item
If for some $\eta > 0$, $\norm{x}_S\le \rho \eta$
for all $x\in\cS_\rho$,
then for all $x\in\cB_\rho$ it holds that
$\norm{x}_S \le \rho \eta/C_M^{1/p}$.
\item
If for some $\eta > 0$, $\norm{x}_S \ge \rho \eta$
for all $x\in \cS_\rho$,
then for all $x\notin\cB_\rho$ it holds that
$\norm{x}_S \ge \rho \eta C_M^{1/p}$.
\end{enumerate}
\end{lemma}

\ifCONF\else
\begin{proof}
For (i), there is some $\kappa \ge 1$ so that $\norm{\kappa x} = \rho$.
We have by hypothesis and from \eqref{eq scale insens gen} that
\[
\norm{x}_S
	\le (C_M\kappa)^{-1/p}\norm{\kappa x}_S
	\le {C_M^{-1/p}} \rho \eta,
\]
and claim (i) follows.

For (ii), there is some $\alpha \le 1$ so that $\norm{\alpha x}= \rho$.
We have by hypothesis and from \eqref{eq scale insens gen}
\[
\norm{x}_S
	\ge (C_M/\alpha)^{1/p} \norm{\alpha x}_S
	\ge C_M^{1/p} \eta \rho,
\]
and claim (ii) follows.
\end{proof}
\fi 

\subsection{Sampling Matrices}

We discussed sampling matrices in \S\ref{subsubsec point red 1}, and their
construction
via well-conditioned bases (Definition~\ref{def:wcb}), using
the fast construction of a change-of-basis matrix $R$ (Theorem~\ref{thm well cond}).

\begin{definition} $T_S$, $\norm{SA}_v$, $\norme{SA}_v$.
Let $T_S\subset [n]$ denote the indices $i$ such that $e_i$ is chosen for $S$.
Using a probability vector $q$ and sampling matrix $S$ from $q$, we will
estimate $\norm{A}_v$ using $S$ and a re-weighted version, $\norm{S\cdot}_{v, w'}$ of $\norm{\cdot}_v$,
with
\[
\norm{SA}_{v,w'} \equiv \left[\sum_{i\in T_S} w'_i M(\norm{A_{i*}})\right]^{1/p},
\]
where $w'_i \equiv w_i/q_i$.
Since $w'$ is generally understood, we will usually just write $\norm{SA}_v$.
We will also need an ``entrywise row-weighted'' version:
\[
\norme{SA}
	\equiv \left[\sum_{i \in T_S} \frac{w_i}{q_i} \norm{A_{i*}}_M^p\right]^{1/p}
	= \left[ \sum_{\substack{i\in T_S\\j\in[d]}} \frac{w_i}{q_i} M(A_{ij})\right]^{1/p}.
\]

\end{definition}

We have $\E_S[\norm{SA}_v^p] = \norm{A}_v^p$.

When $M$ is scale-invariant, we can scale the rows of $S$ by $w'$,
and assume that $w'$ is the vector of all ones.


\begin{lemma}\label{lem eM M}
Let $x\in\R^d$ and $M$ even, monotone, and polynomial.
For weights $w=\bone^d$ we have
\begin{equation}\label{eq M bound}
\frac{1}{d}\norm{x}_M^p \le M(\norm{x}_p) \le \norm{x}_M^p.
\end{equation}
We now let $w$ be general (but as always $w\ge\bone_d$).
For $p\le 2$ we have
\[
\frac1{d^{1/p}}\norme{A} \le \norm{A}_v \le \norme{A},
\]
while for $M(x)=|x|^p$ with $p\ge 2$, we have
\[
\norme{A} \le \norm{A}_v \le d^{1/2-1/p}\norme{A},
\]
and for $M(x)=|x|^p$ with $p\le 2$, we have
\[
d^{1/2-1/p} \norme{A} \le \norm{A}_v \le \norme{A}.
\]
\end{lemma}

\ifCONF\else
\begin{proof}
For the first inequality, from $\norm{x}_p\ge \norm{x}_\infty$
we have by monotonicity
\[
M(\norm{x}_p)\ge M(\norm{x}_\infty)\ge \frac1d\sum_i M(x_i) = \frac1d\norm{x}_M^p,
\]
and for the second inequality,
\begin{equation*}
M(\norm{x}_p)
	\le \frac{\norm{x}_p^p}{\norm{x}_\infty^p} M(\norm{x}_\infty)
	=   \sum_i \frac{|x_i|^p}{\norm{x}_\infty^p}M(\norm{x}_\infty)
	\le \sum_i M(x_i)
	= \norm{x}_{M}^p.
\end{equation*}
The claim regarding matrix norms follows from \eqref{eq M bound}, the definitions,
and standard results.
\end{proof}
\fi 


\begin{lemma}\label{lem M lev}
Consider norms under $w=\bone^n$.
For $M$-estimators with $M()$ even, polynomial, and linearly bounded below,
\[
\sup_{y\in\R^d} \frac{M(A_{i*}y)}{\norm{Ay}_M^p}
	\le \gamma_i(A,M) \equiv\max\{\beta\norm{U_{i*}}_p/C_M,\beta^p\norm{U_{i*}}_p^p\},
\]
where $U$ is an $(\alpha,\beta,p)$-well-conditioned basis of $A$
(see Definition \ref{def:wcb}).
For estimators in $\cL_p$, with $M(x)=|x|^p$ with $p\ge 1$,
$\gamma_i(A, M)$ can be sharpened to $\beta^p\norm{U_{i*}}_p^p$, and
for $\gamma(A,M)\equiv \sum_i \gamma_i(A,M)$ we have
$\gamma(A,M) = O((\alpha\beta)^p) = O(d^{\max\{p,1+p/2\}})$.
For $M\in \cM_2$, a general nice $M$-estimator with $p\le 2$,
$U$ can be an orthogonal basis of $A$, and
$\gamma(A,M) \equiv \sum_i \gamma_i(A,M) = O(\sqrt{{d} n})/C_M$.
\end{lemma}

We will use $\gamma_i(A)$ and $\gamma(A)$ to mean the appropriate bounds for the
particular $M$-estimator under consideration, and write $\gamma_i(A,\cM_2)$,
$\gamma_i(A,\cL_p)$, and so on, for those classes of estimators.

\ifCONF
\begin{IEEEproof}
\else
\begin{proof}
\fi
Let $A=UR$. Let $p' \equiv 1/(1-1/p)$, with $p'\equiv\infty$ when $p=1$.
We have 
\begin{align*}
M(A_{i*}y)
	   & = M(U_{i*}Ry)
	\\ & \le M(\norm{U_{i*}}_p \norm{Ry}_{p'}) &\text{by H\"older's inequality}
	\\ & \le M(\norm{U_{i*}}_p \beta \norm{URy}_{p})& \text{$U$ is well-cond.}
	\\ & \le \max\{\beta\norm{U_{i*}}_p/C_M,\beta^p\norm{U_{i*}}_p^p\} M(\norm{Ay}_p) & \text{growth bounds for $M$}
	\\ & \le \max\{\beta\norm{U_{i*}}_p/C_M, \beta^p\norm{U_{i*}}_p^p\} \norm{Ay}_M^p. & \text{by \eqref{eq M bound}}
\end{align*} 
The first claim of the lemma follows. The claim for $\cL$ follows by noting that stronger bound possible
in the next-to-last inequality.

The bound for $\gamma(A,\cL) \le (\alpha\beta)^p/C_M$ follows using the definition of an
$(\alpha,\beta,p)$-well-conditioned basis
and constructions of such bases (Theorem 5, \cite{DDHKM09}). The bound for
$\gamma(A,\cM_2)$ follows by using an orthogonal basis, 
and the bound on $\gamma_i(A,\cM_2)$
applying in such a case: the worst case has each $\norm{U_{i*}}_2$ equal
to $\sqrt{d/n}$.
\ifCONF
\end{IEEEproof}
\else
\end{proof}
\fi

\begin{lemma}\label{lem M lev w}
Let $w\in \R^n$ with $w\ge \bone^n$. Let $N\equiv \lceil \log_2(1+\norm{w}_\infty)\rceil$.
For $j\in[N]$, let $T_j\equiv \{i\in [n] \mid 2^{j-1} \le w_i < 2^j\}$,
and let $U^j$ be an $(\alpha,\beta,p)$-well-conditioned basis for $A_{T_j *}$,
the matrix comprising the rows $A_{i*}$ with $i\in T_j$.
For $M$-estimators with $M()$ even, polynomial, and linearly bounded below,
\[
\sup_{y\in\R^d} \frac{w_i M(A_{i*}y)}{\norm{Ay}_M^p}
	\le \gamma_i(A,M,w)\equiv 2\max\{\beta\norm{U_{i*}^j}_p/C_M,\beta^p\norm{U_{i*}^j}_p^p\}.
\]
The values $\gamma(A,M,w)\equiv \sum_i \gamma_i(A,M,w)$ are as in Lemma~\ref{lem M lev},
but multiplied by $2N$ if $w\ne \bone^n$.
\end{lemma}

\ifCONF\else
\begin{proof}
For given $i$ with $i\in T_j$,
apply Lemma~\ref{lem M lev} to $A_{T_j*}$ with unit weights.
For the given weights and some $y\in\R^d$,
$\norm{Ay}_{M,w} \ge 2^{j-1}\norm{A_{T_j*}y}_M$, where $w_i \le 2^j$ for all $i\in T_j$.
Thus up to using $U^j$, and
a factor of 2, the same bounds hold as in Lemma~\ref{lem M lev}. The bounds for $\gamma(A,M,w)$ are those
for $\gamma_i(A,M)$, multiplied
by that factor of $2$, and by the $N$ upper bound on the number of nonempty $T_j$.
\end{proof}
\fi 

We will need the following lemma and theorem for fast row norm estimation.

\begin{lemma}\label{lem:kappa}
Let $\kappa \in (0,1)$ be an arbitrary constant. 
Let $G$ be a $d \times t$ matrix of i.i.d. normal random variables with mean $0$ and variance $1/t$,
for some $t$ that is a constant multiple of $1/\kappa$.  
For each $i \in [n]$, let $g_i = \|(AG)_i\|_2^2$. 
Then with failure probability $1/n$,
simultaneously for all $i \in [n]$, $g_i \geq \|A_i\|_2^2/n^{\kappa}$.
\end{lemma}

\ifCONF\else
\begin{proof}
Each entry of $(AG)_i$ is an $N(0, \|A_i\|_2^2/t)$ random variable, and so with probability $1-O(n^{-\kappa})$, 
the entry has value at least $\|A_i\|_2/n^{\kappa/2}$, using that $t$ is a constant. Hence, with this probability,
its square has value at least
$\|A_i\|_2^2/n^{\kappa}$. Since the entries in a row of $AG$ are independent, the probability that all squared
entries are less than $\|A_i\|_2^2/n^{\kappa}$ is less than $1/n^2$, since $G$ has $O(1/\kappa)$ columns.
The lemma follows by a union bound. 
\end{proof}
\fi 

\begin{theorem}\label{thm G est}
Let $t_M\equiv 1$ for $M\in\cL_p$, and $t_M$ be a large enough constant, for $M\in\cM_2$.
Fix integers $r_1,m\ge 1$.
For matrix $U\in\R^{n\times d}$, suppose a sampling matrix $S$ using probabilities
$z_i\equiv\min \{1, r_1 z'_i/\sum_i z'_i\}$, where $z'_i=\norm{U_{i*}}_p^m$,
has small constant failure probability,
for some success criterion. (Here we require that oversampling does not harm success.)
Let $G\in\R^{d\times t_M}$ be a random matrix with independent Gaussian entries
with mean 0 and variance $1/t_M$.
Then for $M\in\cL_p$,
a sampling matrix chosen with probabilities
\[
q_i\equiv\min \{1, K_2 d^{m/2} r_1^{m+1} q'_i/\sum_i q'_i\},
\]
where $q'_i \equiv |U_{i*}G|^m$,
also succeeds with small constant failure probability.
For $M\in\cM_2$, for which $p=2$ and $m=1$,
replace $d^{m/2}r_1^{m+1}$ in the expression for $q_i$
with $r_1n^{O(1/t_M)}\log n$, again with small constant failure probability.
\end{theorem}

\ifCONF\else
\begin{proof}
We will show that sampling with $q$ gives performance comparable to sampling with~$z$.

First consider $t_M=1$ for which $G$ is a $d$-vector $g$.

The quantity $ |U_{i*}g|$ is a half-normal random variable; its mean is
proportional to $\norm{U_{i*}}_2$, which for half-normal distributions
implies that its $m$'th moment is 
proportional to $\norm{U_{i*}}_2^m\ge [d^{-1/2}\norm{U_{i*}}_p]^m$,
so that $\E[q'_i] \ge d^{-m/2} z'_i$. 

By a Markov bound, with failure probability
$1/10$, $\sum_i q'_i \le 10 \E_g[\sum_i q'_i] \le C_1 \sum_i z'_i$,
for a constant $C_1$.  We condition on this event $\cE_Q$.

For value
$s>0$, say that index $i$ is \emph{good} if $q'_i\ge d^{-m/2} z'_i/s^m$.
By standard properties of the Gaussian distribution, there is
an absolute constant $C$ so that
the probability that $i$ is not good is at most $C/s$.
\Ken{should've said what these properties are...}

We have $\sum_i z_i \prob_g\{i\mathrm{\ not\ good}\}\le Cr_1/s$,
and conditioning on $\cE_Q$ increases
this bound by at most $1/(1-1/10)=1.1$.
Let $\cE_f$ be the event that this sum is at most $11Cr_1/s$;
then $\cE_f$, given $\cE_Q$, holds with failure probability $1/10$.

Now condition on both events $\cE_f$ and $\cE_Q$.

Let $s=C_2r_1$, with $C_2\equiv 110C$,
so that an algorithm based on sampling with
$z$ will, with probability at least $9/10$,
choose only good indices. That is,
a $z$-sampling algorithm that is restricted to non-failed indices
will have failure probability at most 1/10 more than
one that is not.

Now consider sampling with $q$, but restricting
the algorithm to good indices. A given
good index $i$ is chosen either with probability $1\ge z_i$,
or with probability
\[
K_2 r_1^{m+1} d^{m/2} q'_i/\sum_i q'_i
	\ge K_2(s/C_2)^m r_1 (z'_i/s^m)/(C_1\cE_g[\sum_i q'_i])
	\ge r_1z'_i/\sum_i z'_i \geq z_i,
\]
for large enough constant $K_2$, so that $q_i \ge z_i$
for a good index $i$. We have that a $q$-sampling algorithm
picks indices with probabilities at least as large
as an algorithm that has failure probability at most
3/10 more than an unrestricted $z$-sampling algorithm.
Since picking additional not-good indices does not hurt the performance
guarantee, the claim for $r_1^{m+1}$ follows.

For $M\in\cM_2$, we invoke Lemma~\ref{lem:kappa},
so that with failure probability $1/n$,
$\norm{U_{i*}G}_2 \ge n^{-\kappa} \norm{U_{i*}}_2$.
With failure probability $1/n$, $\sum_i q'_i$ is 
bounded above by $(\log n)\norm{U}_v$, and so oversampling as
described will have all rows chosen with probability at least $z_i$.
Adjusting constants, the theorem follows.
\end{proof}
\fi 

\begin{lemma}\label{lem S bb}
Let $\rho>0$ and integer $z>0$.
For sampling matrix $S$,
suppose for given $y\in\R^d$ with failure probability $\delta$ it holds
that $\norm{SAy}_M = (1\pm 1/10) \norm{Ay}_M$.
There is $K_1 = O(z^2/C_M)$ so that
with failure probability $\delta (K_\cN/C_M)^{(1+p)d}$,
for a constant $K_\cN$,
any rank-$z$ matrix $X\in\R^{d\times d}$
has the 
property that if $\norm{AX}_v \ge  K_1\rho$, then $\norm{SAX}_v \ge \rho$,
and that if $\norm{AX}_v \le \rho/K_1$, then $\norm{SAX}_v\le \rho$.
\end{lemma}

\ifCONF\else
\begin{proof}
Suppose $X$ has the SVD $X= U\Sigma V^\top =W V^\top$,
where $W= U \Sigma\in\R^{d\times z}$ and $V\in\R^{d\times z}$
has orthogonal columns. 
Since $\norm{AX}_v = \norm{AW}_v$, and similarly for $\norm{SAX}_v$,
it is enough to assume that $\norm{AW}_v \ge K_1\rho$,
and show that $\norm{SAW}_v \ge 10\rho$ follows.

From Lemma~\ref{lem eM M}, when $p\le 2$,
$\norme{AW} \ge \norm{AW}_v$,
so there is a column $y$ of $W$ such that $\norm{Ay}_M\ge K_1\rho/\sqrt{z}$.
The same bound follows by similar reasoning for $p\ge 2$.

We will show that if $\norm{Aw}_M$ is large, then $\norm{SAw}_M$
is large, for all $w\in\R^d$, by applying the results of \S\ref{subsec scale insens},
with $\cV$ of that section mapping to $\rspace(A^\top)$,
$\norm{}$ mapping to a pseudo-norm on vectors $\norm{A\cdot }_M$,
and $\norm{}_S$ similarly mapping to $\norm{SA\cdot }_M$.
(Here we use the fact that $\rspace(A^\top)$ is the orthogonal
complement of the nullspace of $A$, so that $\norm{A\cdot}_M$ is a norm within $\rspace(A^\top)$:
our claims are such that if they hold all members of $\rspace(A^\top)$, then they hold for all $w\in\R^d$.)

Let $\cS_\gamma$ denote the sphere in $\rspace(A^\top)$ with respect to $\norm{A\cdot}_M$
of radius $\gamma\equiv K_1 \rho/z$. 

Let $\cN$ be a $C_M\eps_s^2\gamma$-cover of $\cS_\gamma$,
for $\eps_s=1/10$. From Lemma~\ref{lem v-eps-net gen},
$|\cN| \le (K_\cN/C_M)^{(1+p)m}$, for a constant $K_\cN$.
With failure probability at most $\delta(K_\cN/C_M)^{(1+p)m}$,
all $w\in\cN$
satisfy the condition 
\begin{equation}\label{eq vec good}
\norm{SAw}_M = (1\pm\eps_s)\norm{Aw}_M.
\end{equation}
Assume this holds.

From (ii) of Lemma~\ref{lem v-dilation gen},
this implies $\norm{SAw}_M \ge \gamma/2$ for all $w\in\cS_\gamma$.
From (ii) of Lemma~\ref{lem weak embed gen}, we then have
$\norm{SAw}_M\ge \gamma C_M^{1/p}/2$, for all $w\in\R^m$
outside $\cS_\gamma$. Now applying this to $y$,
we have, again using Lemma~\ref{lem eM M},
\[
\norm{SAX}_v
	= \norm{SAW}_v
	\ge \norme{SAW}/z
	\ge \norm{SAy}_v/z
	\ge (K_1\rho/\sqrt{z})C_M^{1/p}/2z
	= \rho K_1 C_M^{1/p}/2z\sqrt{z}.
\]
So if $K_1 = 2 z^2/C_M^{1/p}$, the assumption that
$\norm{AX}_v\ge K_1\rho$ implies that $\norm{SAX}_v\ge \rho$,
assuming the net condition above and its failure probability bound.

For the upper bound case, consider $p\le 2$.
We have $\norm{AX}_v\le \rho/K_1$
implies $\norme{AW}\le \rho/K_1$,
so every column $y$ of $W$ has $\norm{Ay}_M\le \rho / K_1$.
Let $\gamma\equiv \rho / K_1$.

From \eqref{eq vec good} and  (i) of Lemma~\ref{lem v-dilation gen},
 $\norm{SAw}_M \le \frac32 \gamma$ for all $w\in\cS_\gamma$.
From (i) of Lemma~\ref{lem weak embed gen}, we then have
$\norm{SAw}_M \le \gamma \frac32 C_M^{1/p}$, for all $w\in\R^d$
inside $\cS_\gamma$. We have, again using Lemma~\ref{lem eM M},
\begin{align*}
\norm{SAX}_v
	  & = \norm{SAW}_v
	\\ & \le \norme{SAW}_v
	\\ & \le (z^{1/p} \gamma \frac32 C_M^{1/p})
	\\ & \le z^{1/p}(\rho/K_1 )\frac32 C_M^{1/p}
	\\ & = \rho z^{1/p} \frac32 C_M^{1/p}/K_1.
\end{align*}
So if $K_1 = z^2 \frac32 C_M^{1/p}$, the assumption that
$\norm{AX}_v\le \rho/K_1$ implies that $\norm{SAX}_v\le \rho$,
assuming the net condition above and its failure probability bound. The argument
for $p\ge 2$ follows similarly, and the lemma follows.
\end{proof}
\fi 

\begin{lemma}\label{lem tail bound 2}
With notation as in Lemma~\ref{lem M lev w},
for $r>0$ let $\hat r \equiv r/\gamma(A,M,w)$, and let $q\in\R^n$ have
\[
q_i\equiv \min\{1, \hat r \gamma_i(A,M,w)\}.
\]
Let $S$ be a sampling matrix generated using $q$,
with weights as usual $w'_i = w_i/q_i$.
Let $W\in\R^{d\times z}$, and $\delta>0$. There is an absolute constant $C$ so that for
$\hat r \ge C z \log(1/\delta)/\eps^2$, with failure
probability at most $\delta$,
\[
\norm{SAW}_{v,w'}= (1\pm\eps)\norm{AW}_{v,w}.
\]
\end{lemma}

\ifCONF\else
\begin{proof}
For $T_S$ the set of rows chosen by $S$,
let $Z_i$
be the random variable
$\Ibr{i\in T_S}\frac{w_i}{q_i} M(\norm{A_{i*}W}_2)$,
so that $\E[\sum_i Z_i] = \norm{AW}_v^p$.
We note that terms with $q_i=1$ can be ignored.
Using Lemmas~\ref{lem M lev w} and \ref{lem eM M},
\begin{align*}
\sum_i \Var[Z_i]
	   & \le \sum_i q_i \frac{w_i^2}{q_i^2} M(\norm{A_{i*}W}_2)^2
	\\ & \le \sum_i \frac1{\hat r\gamma_i(A,M,w)} \sum_{j\in [z]} w_i^2 M(A_{i*}W_{*j}) M(\norm{A_{i*}W}_2) 
	\\ & \le \frac1{\hat r} \sum_{j\in [z]} \norm{AW_{*j}}_M^p \sum_i w_i M(\norm{A_{i*}W}_2) & \text{by Lem.~\ref{lem M lev w}}
	\\ & = \frac1{\hat r}\norme{AW}^p\norm{AW}_v^p
	\\ & \le \frac{z}{\hat r}\norm{AW}_v^{2p}.	& \text{by Lem.~\ref{lem eM M}}
\end{align*}
Again using Lemmas~\ref{lem M lev w} and \ref{lem eM M},
\begin{align*}
|Z_i - \E[Z_i]|
	   & = (1-q_i)\frac{w_i}{q_i}M(\norm{A_{i*}W}_2)
	\\ & \le \frac{w_i}{q_i}\sum_{j\in[z]} M(A_{i*}W_{*j})
	\\ & \le \sum_{j\in [z]} \frac{1}{q_i}\sum_{j\in [z]} \norm{AW_{*j}}_M^p \gamma_i(A,M,w)
	\\ & \le \frac{1}{\hat r}\norme{AW}^p
	\\ & \le \frac{z}{\hat r}\norm{AW}_v^p.
\end{align*}
The lemma follows using
Bernstein's inequality, adjusting constants, and taking the $p$'th roots.
\end{proof}
\fi 

\begin{lemma}\label{lem M kinda embed}
Let $\delta, \rho>0$ and integer $z>0$.
For sampling matrix $S$ chosen as in Lemma~\ref{lem tail bound 2}
with
\[
r= O(\gamma(A,M,w) \eps^{-2} z dz\log(z/\eps)\log(1/\delta)),
\]
it holds with failure probability $\delta$ that
any rank-$z$ matrix $X$ with $d$ columns has the 
property that either $\norm{SAX}_v\ge \rho$,
or $\norm{SAX}_v = \norm{AX}_v(1\pm\eps) \pm \eps\rho$.
\end{lemma}

\ifCONF\else
\begin{proof}
We need only consider
$W\in\R^{d\times z}$ where $W=U\Sigma$
and $X$ has the singular value decomposition
$X=U\Sigma V^\top$, since $\norm{AW}_v = \norm{AX}_v$.

We apply Lemma~\ref{lem S bb},
so that with failure probability to be discussed,
$W$ having $\norm{AW}_v\ge K_1\rho$ must have
$\norm{SAW}_v \ge \rho$, one of the conditions
of the lemma; here $K_1 = O(z^2/C_M)$.

So consider $W\in \cB\equiv \{W\in\R^{d\times z}\mid \norm{AW}_v \le K_1\rho\}\cap\rspace_z(A^\top)$,
where $\rspace_z(A^\top)$ comprises $d\times z$ matrices whose columns are in $\rspace(A)$.
Let $\cN$ be an $\frac1{K_1}\eps\rho$-cover of $\cB$
with respect to the norm $\norm{A\cdot}_v$. To bound the size of $\cN$,
we use Lemma~\ref{lem v-eps-net gen}, where $\cV$ of that lemma
maps to $\rspace_z(A^\top)$, $d$ of that lemma maps to $dz$,
$\Vol()$ maps to the volume of a subset of $\rspace_z(A^\top)$,
considered as a subset of $\R^{d\dim(\rspace(A^\top))}$, $\cB_\rho$ of the lemma
maps to $\cB$, and $\eps$ of the lemma maps to $\frac1{K_1}\eps$.
We have
\[
|\cN| \le (4^pC_M^{-1}(K_1/\eps))^p)^{dz} = O((z^2/\eps)^{pdz}).
\]

Now for given $W'\in\cN$, apply Lemma~\ref{lem tail bound 2}
to obtain
\[
\norm{SAW'}_v = (1\pm\eps)\norm{AW'}_v
\]
with failure probability for given $W'$ to be discussed.

Also, for each
$W'\in\cN$, apply Lemma~\ref{lem S bb}, with
$\rho$ of that lemma equal 
to $\eps\rho$, and $X=W-W'$ with rank $z$. It follows
that with failure probability to be discussed,
for given $W'$, if $W$ has $\norm{A(W-W')}_v\le \eps\rho/K_1$,
then $\norm{SA(W-W')}_v\le \eps\rho$.
Assuming this condition, and the condition on $\norm{SAW'}_v$,
\begin{align*}
\norm{SAW}_v
	   & = \norm{SAW'}_v \pm \eps\rho
	\\ & = \norm{AW'}_v(1\pm\eps) \pm \eps\rho
	\\ & =(\norm{AW}_v\pm \eps\rho/K)(1\pm\eps) \pm \eps\rho
	\\ & =\norm{AW}_v(1\pm\eps) \pm \eps\rho, 
\end{align*}
up to a rescaling of $\eps$ by a constant factor. 

It remains to discuss failure probabilities and the sample size for $S$.
We have $|\cN|+1 = O((z^2/\eps)^{pdz})$ applications of Lemma~\ref{lem S bb},
so for $\delta_y$ the failure probability for a given $y\in\R^d$ having
$\norm{SAy}_v = (1\pm1/10)\norm{Ay}_v$, the failure probability
for these events is $\delta_y O(C_M^{-(1+p)d}(z^2/\eps)^{pdz})$.
Letting $\delta_W$ denote the failure probability for
having $\norm{SAW'}_v = (1\pm\eps)\norm{AW'}_v$ for given
$W'$, with $|\cN|$ such $W'$, we have total failure probability
within a constant factor of
\[
(\delta_W + \delta_y C_M^{-(1+p)d}))(z^2/\eps)^{pdz}.
\]
Applying Lemma~\ref{lem tail bound 2} for each case, a sample size
\[
O(\gamma(A,M,w) \eps^{-2} z^2 d \log(z/\eps)\log(1/\delta))
\]
suffices to give failure probability $\delta$. The lemma follows.
\end{proof}
\fi 


\section{Residual sampling for dimensionality reduction} \label{sec non_adapt}

The following theorem is a variation and extension of Theorem~9 of Deshpande and Varadarajan \cite{dv07}.

\def\hX{{\hat X}}

\begin{theorem}\label{thm non adapt det}
Fix $K\ge 2$.
Let $\hat X\in\R^{d\times d}$ be a projection matrix such that $\norm{A(I-\hX)}_v \le K\Delta^*$,
where as usual $\Delta^*\equiv \norm{A(I-X^*)}_v$, with $X^*\equiv \argmin_{\rank X=k}\norm{A(I-X)}_v$.
For value $r$, let $S$ be a sampling matrix built with probability vector $z$ defined by
$z'_i \equiv M(\norm{A_{i*}(I-\hX)}_2)$ and
$z_i = \min\{1, K_2 r z'_i/\sum_i z'_i\}$, for a constant $K_2$ to be determined.
Let $U$ be an orthogonal basis for the linear span of the rows of $SA$ combined with those
of $\hX$.
Then
there is $r=O(K k^{2+p}\eps^{-p-1}\log(k/\eps))$ such that with constant failure
probability, $\min_{\rank X=k}\norm{A(I-U^\top XU)}_v\le (1+\eps)\Delta^*$.
\end{theorem}

The main difference between this theorem and Theorem~9 of \cite{dv07} is that the latter considers a sampling procedure where the rows
of $SA$ are chosen sequentially, and the probability of choosing a row depends on the rows already chosen. However, the 
proof of \cite{dv07} can be adapted to show that the above non-adaptive version gives the same results. Secondarily,
we note that the proof of \cite{dv07} carries through for nice $M$-estimators.

\ifCONF\else
\begin{proof}
We will outline the changes needed to the proof of \cite{dv07}. (So some statements given here are proven
in \cite{dv07}.) Their proof analyzes the situation as the rows of $SA$ are chosen
one by one. We follow this analysis, even though our sample is chosen ``all at once.'' So order the rows of $SA$ arbitrarily,
and let $H_\ell$, $\ell\in [r]$, denote the linear span of the rowspace of $\hat X$ together with
the first $\ell$ rows in this ordering.

The analysis of \cite{dv07} considers $k+1$ phases in the sequence of the $H_\ell$,
where in phase $j$ there exists a rank-$k$ projection $X_j$ such
that:
\begin{enumerate}[(i)]
\item the dimension of $G\equiv \rspace(X_j) \cap H_\ell$ is at least $j$, that is, the $j$ smallest principal angles between  $\rspace(X_j)$ 
	and $H_\ell$ are zero,
\item and also $\norm{A(I-X_j)}_v\le (1+\delta)^j\Delta^*$, for a parameter $\delta \equiv \epsilon/2k$. That is, $X_0\equiv X^*$. 
\end{enumerate}

That is, the cost of $X_j$ gets slowly worse (ii), but $\rspace(X_j)\cap H_\ell$ gets larger in dimension.
The principal angles of $\rspace(X_k)$ with $H_\ell$ are all zero; that is, $X_k$ is contained in $H_\ell$. Moreover,
$\norm{A(I-X_k)}_v\le (1+\delta)^k\Delta^*\le (1+\eps)\Delta^*$, so for $U$ an orthogonal basis for $H_\ell$ at phase $k$,
$UX_kU^\top = X_k$ is an $\eps$-approximate solution, so a solution exists of the form given in the theorem statement.

Let $Y_\ell$ denote the rank-$k$ projection whose rowspace is that of $X_j$, but rotated about $G$ so as to contain the vector in $H_\ell$ realizing
the smallest nonzero principal angle with $X_j$. By containing that vector and still containing $G$, $Y_\ell$ has more than $j$ zero principal angles
with $H_\ell$,
and so can be $X_{j'}$ for $j'>j$ if it satisfies condition (ii). The proof of \cite{dv07}, and in particular Lemma 10,
shows that as long as $Y_\ell$ does not satisfy (ii), with high-enough probability a near-future sample row $A_{\ell' *}$
will have  residual $\norm{A_{\ell'*}(I-Y_\ell)}_2 \ge (1+\delta/2)\Delta^*$. Such a \emph{witness} to  $\norm{A(I-Y_\ell)}_v$ being large
means that the smallest nonzero principal angle of $X_j$ and $H_\ell$ (and so $Y_\ell$) will become smaller at step $\ell'$;
when the angle becomes small enough, $Y_\ell$ will satisfy (ii), and become $X_{j'}$ for $j'>j$.

We note that outside the proof of Lemma~10 of \cite{dv07}, the proof of Theorem 9 of \cite{dv07} relies on (in our terms),
the monotonicity of the $v$-norm, the triangle inequality $\norm{A_1+A_2}_v \le \norm{A_1}_v+\norm{A_2}_v$,
and relations in Euclidean geometry involving angles between subspaces
and distances of single points to other points or to subspaces.

We therefore focus on the proof of Lemma~10 of \cite{dv07}, which has the
key probabilistic claim that if
\begin{equation}\label{eq big resid}
\norm{A(I-Y_\ell)}_v>(1+\delta)\norm{A(I-X_j)}_v,
\end{equation}
then the probability of picking a witness row $A_{\ell'*}$,
having $\norm{A_{\ell'*}(I-Y_\ell)}_2 \ge (1+\delta/2)\Delta^*$, is at least $(\delta/5K)^p$, or equivalently,
\begin{equation}\label{eq heavy W}
\norm{A_{W*}(I-\hX)}_v \ge \frac{\delta}{5K}\norm{A(I-\hX)}_v,
\end{equation}
where $W\subset [n]$ is the set of indices of the witness rows, and $A_{W*}$ denotes the
matrix with those rows. This is Lemma~10 of \cite{dv07} in our notation.

We now prove this, by showing that \eqref{eq big resid} is false assuming that 
\eqref{eq heavy W} is false.

Let $\tX_\ell$ be the matrix projecting onto $H_\ell$. Then for $i\in W$,
and using that all members of $H_\ell$ are closer to $\rspace(Y_\ell)$ than
to $\rspace(X_j)$,
\begin{align*}
\norm{A_{i*}(I-Y_\ell)}_2
	   & \le  \norm{A_{i*}(I-\tX_\ell)}_2 + \norm{A_{i*}\tX_\ell(I-Y_\ell)}_2 & \mathrm{tri.\ ineq}, \rspace(Y_\ell)\subset H_\ell
	\\ & \le \norm{A_{i*}(I-\tX_\ell)}_2 + \norm{A_{i*}\tX_\ell(I-X_j)}_2
	\\ & \le 2\norm{A_{i*}(I-\tX_\ell)}_2 + \norm{A_{i*}(I-X_j)}_2 & \mathrm{tri.\ ineq}
	\\ & \le 2\norm{A_{i*}(I-\hat X)}_2 + \norm{A_{i*}(I-X_j)}_2. & \rspace(\hat X)\subset H_\ell.
\end{align*}
Combining this bound with $\norm{A_{i*}(I-Y_\ell)}_2 \le (1+\delta/2)\Delta^*$ for $i\notin W$,
we have for $i\in [n]$
\[
\norm{A_{i*}(I-Y_\ell)}_2
	\le (1+\delta/2) \norm{A_{i*}(I-X_j)}_2 + \Ibr{i\in W} 2\norm{A_{i*}(I-\hX)}_2,
\]
and so
\begin{align*}
M(\norm{A_{i*}(I-Y_\ell)}_2)^{1/p}
	   & \le M((1+\delta/2)\norm{A_{i*}(I-X_j)}_2 + \Ibr{i\in W}2\norm{A_{i*}(I-\hX)}_2)^{1/p} & M \mathrm{monotone}
	\\ & \le M((1+\delta/2)\norm{A_{i*}(I-X_j)}_2)^{1/p} + \Ibr{i\in W}M(2\norm{A_{i*}(I-\hX)}_2)^{1/p} & M^{1/p} \mathrm{subadd.}
	\\ & \le (1+\delta/2) M(\norm{A_{i*}(I-X_j)}_2)^{1/p} +\Ibr{i\in W} 2 M(\norm{A_{i*}(I-\hX)}_2)^{1/p}. & M \mathrm{poly}
\end{align*}
Using subadditivity of the $\ell_p$ norm,
\begin{align*}
\norm{A(I-Y_\ell)}_v
	     & \le (1+\delta/2)\norm{A(I-X_j)}_v + 2\norm{A_{W*}(I-\hX)}_v
	  \\ & \le (1+\delta/2)\norm{A(I-X_j)}_v + \frac{2\delta}{5K} K\Delta^* & \mathrm{if\ \eqref{eq heavy W}\ false}
	  \\ & \le (1+\delta)\norm{A(I-X_j)}_v,
\end{align*}
contradicting \eqref{eq big resid}. This implies the result of Lemma~10 of \cite{dv07}, and since the sampling distribution
is only used in the proof of that lemma, we have the same claim of Theorem 10 of \cite{dv07}, for the sampling probability vector~$z$.

(Note that this proof has $\hX$ in \eqref{eq heavy W}, and uses the inequalities
$\norm{A_{i*}(I-\tX_\ell)}_2\le \norm{A_{i*}(I-\hX)}_2\le K\Delta^*$;
the analogous statement in Lemma~10 of \cite{dv07} has $\tX_\ell$ instead of $\hX$,
and upper bounds $\norm{A_{i*}(I-\tX_\ell)}_2$ by $K\Delta^*$.
Aside from using general properties of $M()$, there are the only differences from \cite{dv07}. )
\end{proof}
\fi 


The following algorithm makes use of the sampling scheme implied by this theorem, but estimates the norms
of rows of $A(I - \hX)$ using Gaussians, as in Theorem~\ref{thm G est}.

\begin{algorithm}[H]
\caption{$\textsc{DimReduce}(A, k,\hX, \eps, K)$}
\label{alg non adapt}
{\bf Input:} $A \in \R^{n\times d}$, integer $k \geq 1$,
projection matrix $\hat X\in\R^{d\times d}$ given as $WW^\top$ for $W\in\R^{d\times d_W}$ for a value $d_W$,
$K\ge 1$ a quality bound for $\hX$, so that $\norm{A(I-\hX)}_v\le K\Delta^*$\\
{\bf Output:} $U\in\R^{d\times r_M}$ with orthonormal columns, for a parameter $r_M=\poly(k/\eps)$.
\begin{enumerate}
\item For $M\in\cL_p$, let $t_M\gets 1$; for $M\in\cM_2$, let $t_M$ be large enough  in $O(\log n)$;
\item Let  $G\in\R^{d\times t_M}$ have independent Gaussian entries with mean 0 and variance $1/t_M$;
\item Let $r_1$ be a large enough value in $O(K k^{2+p}\eps^{-p-1}\log(k/\eps))$;
\item For $M\in\cL_p$, let $r\gets r_1^{p+1}$, and for $M\in\cM_2$, let $r\gets r_1$;
\item For $i\in [n]$, let $q'_i \gets M(\norm{A_{i*}(I-\hX) G}_2)$;
\item For $i\in [n]$, let $q_i \gets \min\{1, K_2 r q'_i/\sum_i q'_i\}$, for a large enough constant $K_2$;
\item Let $S$ be a sampling matrix for $q$;
\item Return $U$ such that $U^\top$ is an orthogonal basis for the linear span of the rows of $SA$ combined with those of $\hX$.
\end{enumerate}
\end{algorithm}

\begin{theorem}\label{thm non adapt}
Let $K>0$ and $\hat X\in\R^{d\times d}$ be a projection matrix such that $\norm{A(I-\hX)}_v \le K\Delta^*$,
where as usual $\Delta^*\equiv \norm{A(I-X^*)}_v$, with $X^*\equiv \argmin_{\rank X=k}\norm{A(I-X)}_v$.
Here $\hat X$ is given as $WW^\top$, where $W\in\R^{d\times d_W}$.
Then with small constant failure probability,
$\textsc{DimReduce}(A,k,\hX,\eps,K)$ returns $U\in\R^{d\times (d_W+ K \poly(k/\eps))}$ such that
\[
\min_{\rank X=k}\norm{A(I-UXU^\top)}_v\le (1+\eps)\Delta^*.
\]
The running time is $O(\nnz(A) + dK^{2+2p}\poly(k/\eps)d_W)$  for $M\in\cL_p$  and
$O(\nnz(A)\log n+ dK^2\poly(k/\eps)d_W\log n)$ for $M\in\cM_2$.
\end{theorem}

\ifCONF\else
\begin{proof}
The proof is much like that of Theorem~\ref{thm G est}.

The time to compute all $\norm{A_{i*}(I-\hX) G}_2$ is $O((\nnz(A)+ d_W)t_M)$, and the time to compute
$U$ is $O(d)$ times the square of the number of rows of $S$, which is $O(r)$ with high probability,
where $r= (K\poly(k/\eps))^{p+1}$ for $\cL_p$ and $r =K\poly(k/\eps)$ for $\cM_2$.
The running time bound follows.

We will show that sampling with $q$ gives performance comparable to sampling with
the~$z$ of Theorem~\ref{thm non adapt det}.

For $M\in\cM_2$, we note that with $t_M=O(p^2\log n)$,
$\norm{yG}_2=\norm{y}_2(1\pm 1/2p)$ for $n$ vectors $y$ with high probability,
and so the polynomial growth bounds for $M$ imply that $M(\norm{yG}_2)=M(\norm{y}_2)(1\pm 1/2)$,
and so sample size $r_1$ suffices.

For $M\in\cL_p$, with $t_M=1$, $G$ will be a $d$-vector $g$.
The quantity $|A_{i*}(I-\hX)g|$ is a half-normal random variable; its mean is
proportional to $\norm{A_{i*}(I-\hX)}_2$, which for half-normal distributions
implies that its $p$'th power $q'_i$ has expectation
proportional to $z'_i = \norm{A_{i*}(I-\hX)}_2^p$.

By a Markov bound, with failure probability
$1/10$, $\sum_i q'_i \le 10 \E_g[\sum_i q'_i] = C_1\sum_i z'_i$,
for a constant $C_1$.  We condition on this event $\cE_Q$.

For given parameter
$s$, say that index $i$ is \emph{good} if $q'_i\ge z'_i/s^p$.
By standard properties of the Gaussian distribution, there is
an absolute constant $C$ so that
the probability that $i$ is not good is at most $C/s$.

We have $\sum_i z_i \prob_g\{i\mathrm{\ not\ good}\}\le Cr/s$,
and conditioning on $\cE_Q$ increases
this bound by at most $1/(1-1/10)=1.1$.
Let $\cE_f$ be the event that this sum is at most $11Cr/s$;
then $\cE_f$, given $\cE_Q$, holds with failure probability $1/10$.

Now condition on both events $\cE_f$ and $\cE_Q$.

Let $s=C_2r$, with $C_2\equiv 110C$,
so that an algorithm based on sampling with
$z$ will, with probability at least $9/10$,
choose only good indices. That is,
a $z$-sampling algorithm that is restricted to non-failed indices
will have failure probability at most 1/10 more than
one that isn't.

Now consider sampling with $q$, but restricting
the algorithm to good indices. A given
good index $i$ is chosen either with probability $1\ge z_i$,
or with probability
\[
K_2r_1^{p+1}q'_i/\sum_i q'_i
	\ge K_2(s/C_2)^p r (z'_i/s^p)/(C_1\cE_g[\sum_i q'_i])
	\ge rz'_i/\sum_i z'_i = z_i,
\]
for large enough constant $K_2$, so that $q_i \ge z_i$
for a good index $i$. We have that a $q$-sampling algorithm
picks indices with probabilities at least as large
as an algorithm that has failure probability at most
3/10 more than an unrestricted $z$-sampling algorithm.
Since picking not-good indices does not hurt the performance
guarantee, the theorem follows.
\end{proof}
\fi 


\section{Main Algorithms}

\subsection{Approximate Bicriteria Solutions}

We next give an algorithm, described informally in \S\ref{subsubsec point red 1},
for computing a bicriteria solution. The main algorithm is \textsc{ConstApprox}, which
calls \textsc{ConstApproxRecur}. We follow the algorithm with analysis of Theorem~\ref{thm const approx}.

\begin{algorithm}[H]
\caption{$\textsc{ConstApproxRecur}(A, w)$}
\label{alg:M_recur}
{\bf Input:} $A \in \R^{n'\times d'}$, $\hat A\in\R^{n'\times d}$, weight vector $w\in\R^{n'}$\\
{\bf Output:} Matrix $A'\in\R^{P_M \times d'}$ 
for a parameter $P_M$.
\begin{enumerate}
\item If $n'\le P_M$, {\bf return} $\hat A$. 
\item Compute a well conditioned basis of $A$, and leverage scores $q'_i = \gamma_i(A,M,w)$ as in Lemma~\ref{lem M lev w}
\item Let $r$ be a big enough value in $\poly(d')\sum_i q'_i$; if $M\in\cM_2$, let $r\gets Cr\log\log\log n$, for a constant $C$
\item Let $S$ be a sampling matrix for $A$, using probabilities $q_i \gets \min\{1, rq'_i/\sum_i q'_i\}$
\item If $M\in\cL_p$, scale the rows of $S$ by the corresponding $1/q_i^{1/p}$ and set $w'$ to be a vector of ones;
	if $M\in\cM_2$, set $w'_i \gets w_i/q_i$ for each row $i$ in $S$
\item {\bf return} {\sc ConstApproxRecur}$(SA, S\hat A, w')$.
\end{enumerate}
\end{algorithm}

\begin{algorithm}[H]
\caption{$\textsc{ConstApprox}(A, k)$}
\label{alg:M}
{\bf Input:} $A \in \R^{n\times d}$, integer $k \geq 1$\\
{\bf Output:} $\hat X = UU^\top$, where $U\in\R^{d\times P_M}$ with orthonormal columns,
for a parameter $P_M$.
\begin{enumerate}
\item For parameter $m=\poly(k)$, let $R\in\R^{d\times m}$ be a sparse embedding matrix from
Theorem~\ref{thm good colspace sparse} with constant $\eps$
\item $A'\gets \textsc{ConstApproxRecur}(AR, A, \bone^n)$
\item {\bf return} $UU^\top$, where $U^\top$ is an orthonormal basis for the rowspace of $A'$.
\end{enumerate}
\end{algorithm}

\begin{theorem}\label{thm const approx}
Let parameter $P_M = \poly(k)$ for $M\in\cL_p$, and $P_M=\poly(k)\log^3 n$ for $M\in\cM_2$.
With constant probability, the matrix $U$ output by 
{\sc ConstApprox}$(A, k)$ (Algorithm~\ref{alg:M})
has
$$\|A(I-UU^\top)\|_v \leq K\Delta^*,$$
where $K=\poly(k)$ for $M\in\cL_p$ and $K = (\log n)^{O(\log(k)}$ for $M\in\cM_2$,
and as usual $\Delta^*\equiv \min_{\rank X = k }\|A(I-X)\|_v$.
The running time is $O(\nnz(A)+ (n+d)\poly(k))$ with high probability.
\end{theorem}

\ifCONF\else
\begin{proof}
From Theorem~\ref{thm good colspace sparse} using constant $\eps$,
the matrix $X_1 = \argmin_{\rank X= k}\norm{ARX-A}_v$
has $\norm{A(I-X_1)}_v \le (1+3\eps)\Delta^*$.

We first consider $M\in\cL_p$,  for which, since $\gamma(AR,M)=\poly(d')=\poly(k)$,
we can set $P_M$ such that with small constant failure probability,
there are no recursive calls within {\sc ConstApproxRecur}.

Let $X_2 =  \argmin_{\rank X= k}\norm{S(ARX-A)}_v$. (Remember that
when we use $S$, the corresponding weight $w'$ in $\norm{}_v$ is the one constructed for $S$.)
We note that without loss of generality,
the rowspace of $X_2$ lies in the rowspace of $SA$, since otherwise there
is a rank-$k$ projection $Z$ onto the rows of $SA$ with $\norm{S(ARX_2Z - A)}_v$ smaller
than $\norm{S(ARX_2 - A)}_v$. Thus, the rows of $ARX_2$ are all in a $k$-dimensional
subspace of $SA$, and the output $U$ has $\norm{ARUU^\top - A}_v \le \norm{ARX_2-A}_v$.

It remains to show that $ \norm{ARX_2-A}_v$ is within a small factor of $\norm{ARX_1-A}_v$.
From the triangle inequality, for any $Y\in\R^{r\times d}$,
\[
\norm{S(ARY-A)}_v = \norm{S(ARX_1 - A)}_v \pm \norm{S(ARY - ARX_1))}_v.
\]
We can apply Lemma~\ref{lem S bb} to $\norm{AR(Y-X_1)}_v$, mapping
$AR$ to $A$ of the lemma, $d$ to $r$, $z$ to $2k$,
and $\rho$ to $8\Delta^*$, so that if $S$ has the property that,
for given $y\in\R^d$ with failure probability $\delta$ it holds
that $\norm{SAy}_M = (1\pm 1/10) \norm{Ay}_v$,
then there is $K_1 = \poly(k)$ such that with failure probability
at most $\delta \exp(\poly(r))$, for all $Y$ the condition $\norm{AR(Y-X_1)}_v\ge K_1 8 \Delta^*$
implies that $\norm{SAR(Y-X_1)}_v \ge 8\Delta^*$. From Lemma~\ref{lem tail bound 2},
$S$ can be chosen with $m=\poly(r)=\poly(k)$ such that $\delta\exp(\poly(r)) < 1/10$.
Assume this event  $\cE$ holds.

Since $\E[\norm{S(ARX_1-A)}_v] = \norm{ARX_1-A}_v$, with probability
at least $1/2$,
\[
\norm{S(ARX_2-A)}_v \le \norm{S(ARX_1-A)}_v \le 2(1+3\eps)\Delta^* \le 4\Delta^*,
\]
and so $\norm{SAR(X_1-X_2)}_v \le \norm{S(ARX_1-A)}_v + \norm{S(ARX_2-A)}_v \le 8\Delta^*$.
Therefore assuming $\cE$ holds, $\norm{AR(X_1-X_2)}_v \le 8K_1\Delta^*$,
and
\[
\norm{ARX_2-A}_v \le \norm{ARX_1-A}_v + \norm{AR(X_1-X_2)}_v
	\le (2+8K_1)\Delta^*.
\]
This implies that $X_2$, and so the returned $\hX$, are within $K=\poly(k)$
of optimal. The theorem follows, for $M\in\cL_p$.

For nice general $M\in\cM_2$, we have only the bound
$\gamma(AR,M,w) = \poly(k)\sqrt{n}\log(1+\norm{w}_\infty)$, and so there will be recursive calls
in {\sc ConstApproxRecur}.

The expected value
\[
\E_S\norm{w'}_1 = \E_S[\sum_i w_i/q_i] = \sum_i w_i = \norm{w}_1,
\]
so with failure probability $1/\log n$, $\norm{w'}_1 \le \norm{w}_1\log n$,
and at recursive depth $c\le 2\log \log n$, with failure probability at most $(2\log \log n)/\log n$,
$\norm{w}_\infty \le \norm{w}_1\le n(\log n)^c\le n^2$ for large enough $n$.

The random variable $|T_S|$, the number of rows chosen for $S$, is a sum of random variables
with mean $r$, sum of variances $r$,
and maximum value for each variable at most 1, so by Bernstein's inequality $|T_S|$ is within a constant
factor of its expectation, $r$, with failure probability at most $\exp(-r)$.

So with constant
failure probability, the number of rows $n'$ at recursive depth $2\log \log n$ is most
$\poly(k)\log(1+\norm{w}_\infty)^2 n^{1/2^c}\log\log n \le \poly(k)\log^3n$, the promised value of $P_M$.

Due to the multiplication by $\log\log\log n$ in Step 3, via Lemma~\ref{lem tail bound 2},
the failure probability for the sampling approximation bounds is $O(1/\log\log n)$,
or small constant overall.

With a recursive depth at most $2\log\log n$, the blow up in approximation factor is
$\poly(k)^{2\log \log n}$, which is in $(\log n)^{O(\log(k))}$ as $n\rightarrow\infty$.
The quality bound follows.

The running time of the body of \textsc{ConstApprox} is $O(\nnz(A))$, since $R$ is a sparse embedding.
The running time of the body of \textsc{ConstApproxRecur} is $O(n\poly(k))$, and this dominates
the running time for any recursive calls, since $n$ is reduced in size at least geometrically when it is larger
then $P_M$.
The theorem follows.
\end{proof}
\fi 

Note the running time can be made to be $O(\nnz(A))+ O(n\poly(k))$, that is, without any particular dependence on $d$.

\subsection{$\eps$-Approximations}

We now give the main algorithm.

We assume
an algorithm $\textsc{SmallApprox}(\hat A, B, C, w, k,\eps)$ 
that returns an $\eps$-approximate minimizer of $\norm{\hat A XB-C}_v^p$ over rank-$k$ projections $X$,
where the dimensions of $\hat A$, $B$, and $C$ are all $\poly(k/\eps)$ for $M\in\cL_p$,
and in $\poly(k/\eps)\log n$ for $M\in\cM_2$.

Here \textsc{SmallApprox} for $\cL_p$
is given in the proof of Theorem~\ref{thm small approx}, below,
and the reader must provide their own \textsc{SmallApprox} for $\cM_2$. 

First we give and analyze an algorithm for $M\in\cL_p$, then similarly for $\cM_2$.

\begin{algorithm}[H]
\caption{$\textsc{Approx}\cL_p(A, k,\epsilon)$}
\label{alg eps Lp}
{\bf Input:} $A \in \R^{n\times d}$, integer $k \geq 1$, $\epsilon>0$\\
{\bf Output:} $V\in\R^{d\times k}$ with orthonormal columns
\begin{enumerate}
\item Let $\hX \gets \textsc{ConstApprox}(A,k)$ // {\small Alg. \ref{alg:M}, properties in Thm.~\ref{thm const approx} }
\item Let $U \gets \textsc{DimReduce}(A, k, \hX,\eps, K)$  // {\small Alg. \ref{alg non adapt}, properties in Thm.~\ref{thm non adapt}, $K=\poly(k/\eps)$ the quality bound for $\hX$}
\item Let $m$ denote the number of columns of $U$, where $m=\poly(K)=\poly(k/\eps)$
\item Let $S\in\R^{\poly(m/\eps)\times d}$ be a sparse embedding // {\small a lopsided embedding for $(U, A^\top)$ as in Lem.~\ref{lem lopsided sparse}}
\item Compute $q'_i$ as estimates of $\gamma_i(A[S^\top\ U], M)$ via the methods of Theorems~\ref{thm well cond} and \ref{thm G est}
\item Let $r_1\gets\gamma(A[S^\top\ U],M, w)\poly(k/\eps)$ as in Lemma~\ref{lem M kinda embed} for $A[S^\top\ U]$
\item Find sampling matrix $T$ with probabilities
	$q_i = \min\{1, {\hat d}^{p/2} r_1^{c+1} q'_i/\sum_i q'_i\}$, as in Theorem~\ref{thm G est}, where $\hat d$ is the total number of columns
	of $S^\top$ and $U$
\item Rescale $T$ with $1/q_i^{1/p}$; suppose $T$ has $n_T$ entries
\item Let $Z \gets \textsc{SmallApprox}(TAU, U^\top S^\top, TAS^\top, \bone^{n_T},\eps)$,
where $Z=WW^\top$, $W\in\R^{\poly(k/\eps)\times k}$ with orthonormal columns,
and {\bf return} $U W$.
\end{enumerate}
\end{algorithm}

\begin{theorem}\label{thm Approx Lp}
Let $k \geq 1$, $\eps \in (0,1)$, and $1 \leq p = a/b \in [1,2)$ for integer constants $a,b$.
Let parameter $P_M = \poly(k/\eps)$.
Algorithm~\ref{alg eps Lp} takes $O(\nnz(A) + (n+d)\poly(k/\eps) + \exp(\poly(k/\eps)))$ time
to find rank-$k$ projection $X_1=VV^\top$, where with small constant failure probability
$X_1$ is an $\eps$-approximate solution to $\min_{\rank X=k}\norm{A(I-X)}_v$.
\end{theorem}

\ifCONF\else
\begin{proof}
First consider correctness. We make claims that each hold with some failure probability; we will account
for these probabilities, but assume for now that the claims hold.

From Theorem~\ref{thm non adapt}, for $U$ as in Step 2 of \textsc{Approx},
an $\eps$-approximate
solution to
\begin{equation}\label{eq UXU}
\min_{\rank X= k} \norm{A(I-UXU^\top)}_v
\end{equation}
will yield an $\eps$-approximate solution
to $\min_{\rank X=k}\norm{A(I-X)}_v$.
From Lemma~\ref{lem lopsided sparse}, there is $S$ of the given dimensions which
is a lopsided embedding for $(U, A^\top)$ with respect to $\norm{}_h$, so that a solution to
$\norm{A(I-UXU^\top)S^\top}_v$ will by Lemma~\ref{lem lopsided} be an $\eps$-approximate solution
to  \eqref{eq UXU}.

We show that the sampling matrix $T$ chosen in Step 7 of $\textsc{Approx}\cL_p$ has the property
that an $\eps$-approximate minimizer of $\norm{TA(I-UXU^\top)S^\top}_v$ is also one for \eqref{eq UXU}.
We apply Lemma~\ref{lem M kinda embed}, where $S$ of the lemma maps to a matrix $\hat T$,
$A$ of the lemma
maps to $A[S^\top\ U]$, and for the $X$ of the lemma, we are interested in $\poly(k/\eps)\times m$
matrices of the form
$[-I\ SUX]^\top$, whose product with $A[S^\top\ U]$ is $AS^\top - UXU^\top S^\top$, for symmetric $X$. The $d$ and the $z$
of the lemma are therefore $\poly(k/\eps)$, so that the $r$ of the lemma is $\gamma(A[S^\top\ U], M,w)\poly(k/\eps)$,
the $r_1$ of Step 6. With $\rho$ of Lemma~\ref{lem M kinda embed} set to $10\Delta^*$,  we have that the rank-$k$
minimizer $Z$
of $\hat TA(I-UXU^\top)S^\top$ must have cost within additive $\eps 10\Delta^*$ and relative $1+\eps$
of the best possible for  $A(I-UXU^\top)S^\top$, which is within $1+O(\eps)$ of $\Delta^*$.

It remains to show that the estimates $q'_i$ are suitable for choosing a sample, using the given
expected sample size.
This follows from Theorem~\ref{thm G est}.
The chosen $T$ therefore preserves approximate
solutions.

With small constant failure probability, $n_T=\poly(k/\eps)$,
because $\gamma(A[S^\top\ U],M) = \poly(k/\eps)$.

There are only a constant number of events, and each hold with small (enough)
constant failure probability. Therefore the failure probability for $\cL_p$ is small constant.

The running time of \textsc{ConstApprox} is $O(\nnz(A)+ n\poly(k))$, from Theorem~\ref{thm const approx}.
The running time of \textsc{DimReduce} is $O(\nnz(A) + dK^{2+2p}\poly(k/\eps)d_W)$,
from Theorem~\ref{thm non adapt};
here the approximation factor $K$ from \textsc{ConstApprox} is $\poly(k)$, and $d_W=\poly(k)$ also,
so the running time of \textsc{DimReduce} is $O(\nnz(A) + d\poly(k/\eps))$.

Since $U$ has $\poly(K)=\poly(k/\eps)$ rows, $S$ has $\poly(K/\eps)=\poly(k/\eps)$ rows. Computation of the change-of-basis matrix
$R^{-1}$ for $q'_i$ takes $O(\nnz(A)+\poly(\hat d/\eps)) = O(\nnz(A)+\poly(K/\eps))$ time via Theorem~\ref{thm well cond},
where $H$ of the theorem maps to $[S^\top\ U]$. Computation of the estimates $q'_i$ takes
$O(\nnz(A) + n\poly(K/\eps))$ time. 

Therefore,
except for the final call to \textsc{SmallApprox}, the total time is $O(\nnz(A)+(n+d)\poly(k/\eps))$, as claimed.

\end{proof}
\fi 

\begin{algorithm}[H]
\caption{$\textsc{ApproxRecur}\cM_2(A, U, S, w, k, \epsilon)$}
\label{alg eps recur}
{\bf Input:} $A \in \R^{n'\times d}$, weight vector $w\in\R^{n'}$, $S^\top, U\in\R^{d\times\poly(K/\eps)}$ with $\hat d=\poly(K/\eps)$ total columns, $U$ with orthonormal columns; here $K=(\log n)^{O(\log k)}$\\
{\bf Output:} $V\in\R^{d\times k}$ with orthonormal columns
\begin{enumerate}
\item if $n'\le P_M$ for a parameter $P_M$: let $Z \gets \textsc{SmallApprox}(AU, U^\top S^\top, AS^\top, w, k,\eps)$,
where $Z=WW^\top$, $W\in\R^{\poly(K/\eps)\times k}$ with orthonormal columns,
and {\bf return} $U W$.
\item Compute $q'_i$ as estimates of $\gamma_i(A[S^\top\ U], M,w)$ via the methods of Theorems~\ref{thm well cond} and \ref{thm G est}
\item Let $r_1\gets\gamma(A[S^\top\ U],M, w)\poly(K/\eps)$ as in Lemma~\ref{lem M kinda embed} for $A[S^\top\ U]$
\item Find sampling matrix $T$ with probabilities
	$q_i = \min\{1,  (n')^\kappa(\log n')(\log \log n') r_1 q'_i/\sum_i q'_i\}$, as in Theorem~\ref{thm G est},
	and associated weights $w'$
\item {\bf return} $\textsc{ApproxRecur}(TA, U, S, w', k,\epsilon)$.
\end{enumerate}
\end{algorithm}

\begin{algorithm}[H]
\caption{$\textsc{Approx}\cM_2(A, k,\epsilon)$}
\label{alg eps M2}
{\bf Input:} $A \in \R^{n\times d}$, integer $k \geq 1$, $\epsilon>0$\\
{\bf Output:} $V\in\R^{d\times k}$ with orthonormal columns
\begin{enumerate}
\item Let $\hX \gets \textsc{ConstApprox}(A,k)$ // {\small Alg. \ref{alg:M}, properties in Thm.~\ref{thm const approx}}
\item Let $U \gets \textsc{DimReduce}(A, k, \hX,\eps, K)$  // {\small See Alg. \ref{alg non adapt} and Thm.~\ref{thm non adapt}; $K$ the quality bound for $\hX$}
\item Let $m$ denote the number of columns of $U$, where $m= d_W + K\poly(k/\eps)$ // {\small where $\hat X=WW^\top$, $W\in\R^{d\times d_W}$}
\item Let $S\in\R^{\poly(m/\eps)\times d}$ be a sparse embedding //{\small a lopsided embedding for $(U, A^\top)$ as in Lem.~\ref{lem lopsided sparse} }
\item {\bf return} $\textsc{ApproxRecur}(A, U, S, \bone^n, k, \eps/C\log\log n)$.
\end{enumerate}
\end{algorithm}

\begin{theorem}\label{thm Approx M2}
Let $k \geq 1$, $\eps \in (0,1)$, and $1 \leq p = a/b \in [1,2)$ for integer constants $a,b$.
For a value $K=(\log n)^{O(\log k)}$, let parameter $P_M=\poly(K/\eps)\log^3n $.
Up to calls to \textsc{SmallApprox},
Algorithm~\ref{alg eps M2} takes $O(\nnz(A)\log n + (n+d) \poly(K/\eps) )$ time
to find rank-$k$ projection $X_1=VV^\top$, where with small constant failure probability
$X_1$ is an $\eps$-approximate solution to $\min_{\rank X=k}\norm{A(I-X)}_v$.
\end{theorem}

\ifCONF\else
\begin{proof}
The correctness argument is the same as in Theorem~\ref{thm Approx Lp}, up to the sample size values used.
The estimates $q'_i$ are suitable for choosing a sample, using the given
expected sample size $\gamma(A[S^\top\ U],M, w)\poly(K/\eps)(n')^\kappa(\log n')(\log \log n')$.
This follows from standard random projections results.
The chosen $T$ therefore preserves approximate solutions.

Turning to running time: similar to the proof of Theorem~\ref{thm const approx}, the size
at the first recursive call is
\[
n' = \poly(K k/\eps)n^{1/2+\kappa}\log n \log(1+\norm{w}_\infty)
\]
with constant failure probability, and so the recursive
depth is $O(\log\log n)$ (to get to $P_M=\poly(k/\eps)\log^3n$),
and relative errors $O(\eps/\log\log n)$ yield
total relative error $\eps$. The inclusion of the $\log \log n'$ term in the sample size implies via Lemma~\ref{lem tail bound 2}
that the sampling failure probability per step is $O(1/\log n)$. Since the failure probability for estimation of $q'_i$
is $1/n$ at each call, via Theorem~\ref{thm G est}, the total failure probability is a small constant.
The remainder of the correctness analysis is the same as for $\cL_p$.

The running time of \textsc{ConstApprox} is $O(\nnz(A)+ n\poly(k))$, from Theorem~\ref{thm const approx}.
The running time of \textsc{DimReduce} is $O(\nnz(A)\log n+ dK^2\poly(k/\eps)d_W\log n)$,
from Theorem~\ref{thm non adapt};
here the approximation factor $K$ from \textsc{ConstApprox} is $\poly(k)^{\log\log n} = (\log n)^{O(\log k)}$, and $d_W=\poly(k)\log^3n$,
so the running time of \textsc{DimReduce} is $O(\nnz(A)\log n + d\poly(K/\eps))$. Note also that the number of columns $m$ of $U$
is $d_W+K\poly(k/\eps) = \poly(K/\eps)$.

Since $U$ has $\poly(K/\eps)$ rows, $S$ has $\poly(K/\eps)$ rows. Computation of the change-of-basis matrix
$R^{-1}$ for $q'_i$ takes $O(\nnz(A)+\poly(\hat d/\eps)) = O(\nnz(A)+\poly(K/\eps))$ time via Theorem~\ref{thm well cond},
where $H$ of the theorem maps to $[S^\top\ U]$. Computation of the estimates $q'_i$ takes
$O(\nnz(A) + n\poly(K/\eps))$ time. Thus the body of \text{ApproxRecur} takes $O(\nnz(A) + n\poly(K/\eps))$ time.

With recursive depth $O(\log\log n)$, the total work up to \textsc{SmallApprox}
for \textsc{ApproxRecur} is $O(\nnz(A)+ n \cdot \poly(K/\eps)\log \log n)$. Adding this to the
time $O(\nnz(A)\log n + (n+d)\poly(K/\eps))$ gives the claimed time.
\end{proof}
\fi 


\section{Algorithm for Small Problems}\label{sec basu}

In this section, for $M\in\cL_p$ with $p$ rational,
we show how to find a rank-$k$ subspace which is a $(1+\eps)$-approximation.
We will apply 
a simplified form of Theorem 3 of Basu, Pollack, and Roy \cite{bpr96}.
\begin{theorem}(\cite{bpr96})\label{thm:basu}
Given a set $K = \{\beta_1, \ldots, \beta_s\}$ of $s$ polynomials each of degree at most $d$ in $k$ variables 
with coefficients in $\mathbb{R}$, the problem of deciding whether there exist $X_1, \ldots, X_k \in \mathbb{R}$ for which 
$\beta_i(X_1, \ldots, X_k) \geq 0$ for all $i\in [s]$
can be solved deterministically with $(sd)^{O(k)}$ 
arithmetic operations over $\mathbb{R}$.
\end{theorem}

\begin{theorem}\label{thm small approx}
Assume $p=a/b$ for integer constants $a,b\ge 1$, and let $\eps\in (0,1)$, and integer $k\in[0,m]$.
Given $A\in\R^{m'\times m}$, $B\in\R^{m\times m''}$, and $C\in\R^{m'\times m''}$, with $m',m'' = \poly(m/\eps)$,
a rank-$k$ projection matrix $X$ can be found that minimizes $\norm{AXB-C}_v^p$ up a $(1+\eps)$-factor,
in time $\exp(\poly(m/\eps))$.
\end{theorem}

\ifCONF\else
\begin{proof}
We use Theorem \ref{thm:basu}. Write $X = WW^T$, where $W$ is an $m \times k$ matrix with orthonormal columns.
We think of the entries of $W$ as being variables.
We add the quadratic and linear constraints enforcing the columns of $W$ to be orthonormal. 

Let $D_i$ be the $i$-th row of $AWW^TB - C$. Then $\|AXB - C\|_v^p = \sum_{i=1}^{t'} \|D_i\|_2^p$. 
The entries of $D_i$ are a quadratic polynomial in the entries of $W$. Let $d_i = \|D_i\|_2^{2a}$, which is a degree at most $4a$
polynomial in the entries of $W$ (recall that $p = a/b$). We introduce a variable $e_i$ for each $i \in [m]$, 
with the constraint that $e_i^{2b} = d_i$ and $e_i \geq 0$. 

Note that since $e_i = d_i^{1/(2b)} = \|D_i\|_2^{(2a)/(2b)} = \|D_i\|_2^p$, 
our objective is to minimize $\sum_{i=1}^m e_i$, which is a linear function in the 
$e_i$ variables. The total number of variables in our system is $mk + 2m$, to specify the entries of $W$, and the $d_i$ and $e_i$
for $i \in [m]$. Each polynomial in the system
is of degree $O(1)$, assuming $p = a/b$ and $a,b$ are integer constants, 
and the coefficients can be described using $\poly(nd)$ bits assuming the coefficients of $A$ have this 
property (note that the coefficients of $S$ are in $\{0,1,-1\}$, while the coefficients of $T$ are sampling probabilities
which can be rounded to the nearest power of $2$, and dropped if they are less than $1/n^2$, as otherwise the corresponding row will not be
sampled whp). The total number of polynomial constraints is $O(mk + k^2)$. 

We can minimize $\sum_{i=1}^m e_i$ by performing a binary search. If the cost of the objective function is non-zero, then
using that $\|x\|_2 \leq \|x\|_p$ for $p \leq 2$, while $\|x\|_2 \geq \frac{1}{\poly(d)} \|x\|_p$ for constant $p > 2$, we have
that for rank-$k$ matrices $X$, $\|AXB-C\|_v^p \geq \frac{1}{\poly(d)} (\sigma_{k+1}(C))^{p/2}$, where $\sigma_{k+1}(C)$
is the $(k+1)$-st singular value of $C$. It is known that for an $n \times d$ matrix $V$ with entries specified by $\poly(nd)$
bits, it holds that $\sigma_{k+1}(V) \geq \left (\frac{1}{\exp(\poly(nd))} \right )^k$ if $V$ has rank larger than $k$; see inline (10)
in the proof of Lemma 4.1 of \cite{CW09}. It follows that we can do a binary search using $\poly(m/\eps)$ steps. 

By applying Theorem \ref{thm:basu} $\poly(m/\eps)$ times, once for each step in the binary search, 
we can solve the problem $\min_{\textrm{rank-k projections} X} \|AXB-C\|_v^p$ up to
a $(1+\eps)$-factor given $A$, $B$ and $C$, in time $\exp(\poly(m/\eps))$. 
\end{proof}
\fi 

\begin{remark}
We note that the techniques in this section may also apply more generally to $M$-Estimators. For instance, for the Huber loss function,
it is piecewise polynomial so we could introduce variables for each of the pieces. However, at the moment we reduce an
instance of the $M$-Estimator problem from $n$ points to, at best, $\poly(k\eps^{-1}\log n)$ points, and this $\poly(\log n)$ is problematic when trying
to apply the above ideas since the algorithm is exponential in it.
\end{remark}


\section{Hardness}

Let $E = \{e_i \mid i \in [d]\}$ be the set of vertices of the standard $(d-1)$-simplex, where $e_i$ is the standard unit vector in the $i$-th direction. 
Fix a $k \in [d]$ and a real number $p \in [1, 2)$ and $p$ independent of $k$ and $d$. 
Let $V$ be a $k$-dimensional subspace of $\mathbb{R}^d$, represented as the column span of a 
$d \times k$ matrix with orthonormal columns. 
We abuse notation and let $V$ be both the subspace and the corresponding matrix. 
For a set $Q$ of points, let 
$$c(Q, V) = \sum_{q \in Q} d(q, V)^p = \sum_{q \in Q} \|q^T(I-VV^T)\|_2^p = \sum_{q \in Q} (\|q\|^2 - \|q^TV\|^2)^{p/2},$$
be the sum of $p$-th powers of distances of points in $Q$, i.e., $\|Q-QVV^T\|_v$ with associated $M(x) = |x|^p$. 

\begin{lemma}\label{lem:simplex}
For any $k \in \{1, 2, \ldots, d\}$, the $k$-dimensional subspaces $V$ which minimize $c(E, V)$ 
are exactly the $\binom{n}{k}$ subspaces formed by taking the span of $k$ distinct standard unit 
vectors $e_i$, $i \in [d]$. The cost of any such $V$ is $d-k$. 
\end{lemma}

\ifCONF
\begin{IEEEproof}
\else
\begin{proof}
\fi
By definition of $E$, $c(E, V) = \sum_{i \in [d]} (1 - \|V_{i,*}\|_2^2)^{p/2}$, where $V_{i,*}$ is
the $i$-th row of $V$. Make the change of variables $b_i = 1-\|V_{i,*}\|_2^2$. Then
$\sum_{i=1}^d b_i = d - k$, using that $\|V\|_F^2 = k$. 

Consider the optimization problem
$\min_b \sum_{i=1}^d b_i^{p/2}$ subject to $\|b\|_1 = d-k$ and $b_i \in [0,1]$ for all $i$. Since 
\begin{eqnarray}\label{eqn:ineq} 
(x-\delta)^t + (y+\delta)^t < x^t + y^t,
\end{eqnarray}
for any $x \leq y$ with $x,y \in [0,1]$, $\delta > 0$, and $x-\delta, y+\delta \in [0,1]$, for any $t \in (0, 1)$, it follows that the minimizer $b$ has $k$ coordinates equal to $0$ and remaining $d-k$ coordinates equal to $1$. There are $\binom{d}{k}$ such $b$, all realizing the same minimum value of $\sum_{i=1}^d b_i^{p/2} = d-k$. Further, any other vector $b$ has a strictly larger value of $\sum_{i=1}^d b_i^{p/2}$, since one can find two coordinates $0 < b_i < b_{i'} < 1$ and use (\ref{eqn:ineq}) to find a vector $b'$ with $\|b'\|_1 = d-k$ and for which $\sum_{i=1}^d (b'_i)^{p/2} < \sum_{i=1}^d b_i^{p/2}$. 

Finally, note that we seek to solve this optimization problem, subject to the additional constraints that $b_i = 1 - \|V_{i,*}\|_2^2$. Therefore, the minimum value of our optimization problem is at least $d-k$. Note that each of the $\binom{n}{k}$ subspaces formed by taking the span of $k$ distinct standard unit vectors $e_i$, $i \in [d]$, satisfies that $b_i = 1$ for exactly $d-k$ values of $i$, and $b_i = 0$ otherwise. Therefore, each of these $\binom{n}{k}$ subspaces has the minimum objective function value. 
\ifCONF
\end{IEEEproof}
\else
\end{proof}
\fi

\begin{theorem}\label{thm:main:hardness}
Given a set $Q$ of $\poly(d)$ points in $\mathbb{R}^d$, for a sufficiently small $\eps = 1/\poly(d)$, 
it is NP-hard to output a $k$-dimensional subspace $V$ of $\mathbb{R}^d$ for which 
$c(Q, V) \leq (1+\eps) c(Q, V^*)$, where $V^*$ is the $k$-dimensional subspace minimizing 
the expression $c(Q, V)$, that is, $c(Q,V) \geq c(Q, V^*)$ for all $k$-dimensional subspaces $V$. 
%
\end{theorem}

\ifCONF
\begin{IEEEproof}
\else
\begin{proof}
\fi
Let $G$ be the input graph to the {\sf Clique} problem on $d$ vertices, in which the goal is to determine if $G$ contains a clique of size at least $k$. We assume that $G$ is a {\it regular} graph, and let $r$ be the degree of each vertex. Note that there is a value of $r$, as a function of $d$, 
for which the problem is still NP-hard. Indeed, Garey and Johnson \cite{gj79} 
show that it is NP-hard to find the size of the maximum independent set
even in $3$-regular graphs. As the maximum independent set is the largest clique in the 
complement graph $G$, which is $(d-3)$-regular, there is at least one value of $r$, as a function of $d$, 
for which the {\sf Clique} problem in $r$-regular graphs is NP-hard.

Let $B_1 = \poly(d)$ be sufficiently large and to be specified below. Let $c$ be such that 
\begin{eqnarray}\label{eqn:normalized}
(1-1/B_1)^2 + c^2/B_1 = 1.
\end{eqnarray}
Noting that $(1-1/B_1)^2 = 1 - 2/B_1 + O(1/B_1)^2$, we see that 
$c = \sqrt{2-1/B_1} = \sqrt{2} - O(1/B_1)$, so $c \in (1,2)$. 

We construct a $d \times d$ matrix $A$ as follows: for all $i \in [d]$, $A_{i,i} = 1-1/B_1$, while for $i \neq j$, we have $A_{i,j} = A_{j,i} = c/\sqrt{B_1 r}$ if $\{i,j\}$ is an edge of $G$, and $A_{i,j} = A_{j,i} = 0$ otherwise. Recall here that $G$ is an $r$-regular graph. 
By (\ref{eqn:normalized}), for each row $A_i$ of $A$, we have
$\|A_i\|_2^2 = \left (1- \frac{1}{B_1} \right)^2 + r \cdot \frac{c^2}{B_1 r} = 1.$
Let $B_2 = \poly(d)$ be sufficiently large and to be specified below. Our input set $Q$ to the problem of minimizing $c(Q,V)$ over $k$-dimensional spaces $V$ consists of $B_2$ copies of the $d$ points in $E$, together with the rows of $A$. Notice that all input points have norm $1$. 

We note that in our instance $n = \poly(d)$. If one is interested in achieving hardness for $n = O(d)$, 
one can set $d'= d^{\gamma}$ for a small constant $\gamma > 0$, so that $\poly(d') = n$, and 
use $d'$ in place of $d$ in what follows. The remaining $d-d'$ coordinates in each input 
point are set to $0$. 

Let $\mathcal{W}$ be the set of $\binom{n}{k}$ $k$-dimensional subspaces $V$ of $\mathbb{R}^d$ formed by the span of $k$ distinct standard unit vectors $e_i$, $i \in \{1, 2, \ldots, d\}$. We first consider the cost $c(Q, V)$ for $V \in \mathcal{W}$. 

We identify $V$ with the set $S$ of $k$ coordinates $i$ for which $e_i$ is in the span of $V$. Consider the $i$-th row of $A$. Then we have
$$d(A_i, V)^p = (\|A_i\|_2^2 - \|A_i^S\|_2^2)^{p/2} = (1 - \|A_i^S\|_2^2)^{p/2},$$
where $A_i^S$ denotes the vector which agrees with $A_i$ on coordinates in $S$, and is $0$ otherwise. Indeed,
this follows from the fact that the vectors in the span of $V$ have arbitrary values on the coordinates
in $S$, and have the value $0$ on coordinates outside of the set $S$. Therefore,
$d(A_i, V) = (\|A_i\|_2^2 - \|A_i^S\|_2^2)^{1/2}$. 

First, suppose $i \in S$. Then 
\begin{eqnarray}\label{eqn:inS}
\|A_i^S\|_2^2 = (1-1/B_1)^2 + e(i, S) c^2/(B_1 r),
\end{eqnarray}
where $e(i,S)$
denotes the number of edges in $G$ with one endpoint equal to $i$ and the other endpoint in $S$. So, in this case, 
\begin{eqnarray*}
d(A_i, V)^p & = & (1 - (1-1/B_1)^2 - e(i,S) c^2/(B_1 r))^{p/2}
= (2/B_1 - e(i,S) c^2/(B_1 r) - 1/B_1^2)^{p/2}\\
& = & (1/B_1^{p/2}) (2 - e(i,S) c^2/r - 1/B_1)^{p/2},
\end{eqnarray*}
where the $1/B_1$ term can be made arbitrarily small by making $B_1$ a sufficiently large value of $\poly(d)$. 

On the other hand, suppose $i \notin S$. Then 
\begin{eqnarray}\label{eqn:notinS}
\|A_i^S\|_2^2 & = & e(i,S) c^2/(B_1 r).
\end{eqnarray}
So in this case
\begin{eqnarray*}
d(A_i, V)^p & = & (1 - e(i,S) c^2/(B_1 r))^{p/2}
= 1- e(i,S) c^2 p / (2B_1 r) + O(e(i,S)^2 / (B_1^2 r^2))
= 1 - O(1/B_1),
\end{eqnarray*}
provided $B_1$ is a sufficiently large value of $\poly(d)$, and using that $e(i,S) \leq r$. 

We thus have
\begin{eqnarray}\label{eqn:tryAgain}
c(A, V) & = & (d-|S|) - O(d/B_1) + (1/B_1^{p/2}) \sum_{i \in S} (2 - e(i,S) c^2/r - 1/B_1)^{p/2}.
\end{eqnarray}
We note that we can assume $k \leq r$. Indeed, trivially $k \leq r+1$, while 
determining if there is a
clique of size $r$ in an $r$-regular graph can be solved in polynomial time
by checking, if for each vertex $i$, if $r-1$ of the $r$ neighbors of $i$, together
with vertex $i$, form a clique of size $r$. If a clique is found, then the maximum
clique size of the graph is at least $r$. On the other hand, if the maximum clique size
is at least $r$, and we choose $i$ to be a vertex in the maximum clique, it must be
that $r-1$ of its $r$ neighbors form a clique. Since $\binom{r}{r-1} = r$, the time
complexity is polynomial.  

Since $k \leq r$ and $|S| = k$, 
we have that for $i \in S$, $e(i,S) \leq k-1 \leq r-1$, and so
$2- e(i,S) c^2/r \geq 2/r - 1/B_1 \geq 1/r$, where the final inequality follows
for $B_1$ a large enough $\poly(d)$. Hence, for $i \in S$, 
\begin{eqnarray*}
(2- e(i,S) c^2/r - 1/B_1)^{p/2} & = & (2-e(i,S)c^2/r)^{p/2}(1-1/(B_1(2-e(i,S)c^2/r)))^{p/2}\\
& \geq & (2-e(i,S)c^2/r)^{p/2}(1-r/B_1)^{p/2}\\
& \geq & (2-e(i,S)c^2/r)^{p/2}(1-r/B_1)\\
& = & (2-e(i,S) 2/r)^{p/2} - O(r/B_1),
\end{eqnarray*}
where the first inequality uses that $2-e(i,S)c^2/r \geq 1/r$, and the second inequality
uses that $p \leq 2$ and $B_1 > r$. Note also
$(2- e(i,S) c^2/r - 1/B_1)^{p/2} \leq (2- e(i,S) c^2/r)^{p/2}.$
Plugging into (\ref{eqn:tryAgain}), we have the equality for some value in $O(d^2/B_1)$:
\begin{equation}\label{eqn:try2}
c(A,V)  =  d- |S| - O(d^2/B_1) + (2/B_1)^{p/2} \sum_{i \in S} (1-e(i,S)/r)^{p/2}.  
\end{equation}
If there is a clique of size $k$, then {\it there exists} a $V \in \mathcal{W}$ for which  
(\ref{eqn:try2}) is at most
\begin{equation}\label{eqn:cliqueThere}
c(A,V)  \leq  d- |S| + (2/B_1)^{p/2} |S| (1-(k-1)/r)^{p/2},
\end{equation}
since we can choose $S$ to be the set of coordinates corresponding to those in the clique,
and so each vertex $i \in S$ is incident to $k-1$ other vertices in $S$, so $e(i,S) = k-1$
for all $i \in S$. 

On the other hand, 
if there is no clique of size $k$, then (\ref{eqn:try2}) implies for any $V \in \mathcal{W}$, $c(A,V)$ is at least
\begin{align}\label{eqn:cliqueNotThere}
d- |S| + (2/B_1)^{p/2} (|S|-1) (1-(k-1)/r)^{p/2} + (2/B_1)^{p/2} (1-(k-2)/r)^{p/2} - O(d^2/B_1),
\end{align}
since for any choice of $S$ of $k$ coordinates, we cannot have $e(i,S) = k-1$ for all
$i \in S$, as otherwise the corresponding vertices would constitute a clique of size $k$.
It follows that for at least one $i \in S$, $e(i,S) \leq k-2$ (note also that $e(i,S)$
is at most $k-1$ for all $i \in S$). 

Note that $c(A,V)$ is an additive $(2/B_1)^{p/2} ((1-(k-2)/r)^{p/2} - (1-(k-1)/r)^{p/2}) - O(d^2/B_1)$ larger
in (\ref{eqn:cliqueNotThere}) than in (\ref{eqn:cliqueThere}). Note that using $p \in [1,2)$ 
and $0 \leq k-1 < r$,
\begin{align*}
(1-(k-2)/r)^{p/2} - (1-(k-1)/r)^{p/2} & =  (1-(k-1)/r)^{p/2} ((1 + 1/(r(1-(k-1)/r)))^{p/2} - 1)\\
& =  (1-(k-1)/r)^{p/2} ((1+1/(r-(k-1)))^{p/2} - 1)\\
& \geq  (1/r)^{p/2} ((1+1/r)^{p/2} - 1)
\geq  (1/r) ((1+1/r)^{1/2}-1)
=  \Omega(1/r^2),
\end{align*}
where in the last line 
we used that $1+1/r \geq (1 + 1/(2r) - 1/(8r^2))^2$ which can be verified by expanding
the square. 
Hence, $c(A,V)$ is an additive $\Omega((1/B_1)^{p/2}/r^2) - O(d^2/B_1)$ larger in (\ref{eqn:cliqueNotThere})
than in (\ref{eqn:cliqueThere}). For $B_1$ a large enough $\poly(d)$, and using that $p$ is a constant less than $2$,
this is $\Omega((1/B_1)^{p/2}/r^2)$. Note that in both cases (whether or not $G$ has a
clique of size at least $k$), the cost $c(E,V)$ is $d-k$ for $V \in \mathcal{W}$, as promised
by Lemma \ref{lem:simplex}. Thus, the $B_2$ copies of the $d$ points in $E$ preserve
the additive difference in the two cases. We will show below, in the ``Wrapup'', why
an additive difference of $\Omega((1/B_1)^{p/2}/r^2)$ suffices to complete the proof. 

It remains to handle the case that $V$ is not in $\mathcal{W}$. We have the following lemma.
\begin{lemma}\label{lem:perturb}
Suppose we order the values $\|V_{j,*}\|_2^2$ for $j \in [d]$, and let $v_i$
be the $i$-th largest value in this ordering. Then for
the simplex $E$, $c(E, V) \geq (d-k) + ((1-v_{k+1})^{p/2} + v_{k+1}^{p/2} - 1)$.  
\end{lemma}

\ifCONF
\begin{IEEEproof}
\else
\begin{proof}
\fi
We have $c(E, V) = \sum_{i \in [d]} (1- v_i)^{p/2}$. Making the change of variables $b_i = 1-v_i$,
we have $c(E,V) = \sum_{i \in [d]} b_i^{p/2}$, where $b_i \in [0,1]$ since $v_i \in [0,1]$ for all $i$. 
Also, $\sum_{i \in [d]} b_i = d-k$. If $v_{k+1} = 0$, then the result now follows, as these constraints
imply $v_1 = v_2 = \cdots = v_k = 1$. Otherwise, suppose $v_{k+1} > 0$. Note that under this ordering,
$1 \geq b_n \geq b_{n-1} \geq \cdots \geq b_1 \geq 0$. If there exists a $j > k+1$ for which $b_j < 1$,
then necesarily there is a $j' \leq k$ for which $b_{j'} > 0$, as otherwise we could not have
$\sum_{i=1}^d b_i = d-k$. Then applying (\ref{eqn:ineq}) to $b_{j'}$ and $b_j$ by adding $\delta > 0$ to $b_{j}$
and subtracting $\delta$ from $b_{j'}$, for a sufficiently small $\delta > 0$, $c(E, V)$ can only get
smaller. Repeating this process, we obtain $b_1 = b_2 = \cdots = b_{k-1} = 0$, $b_k = v_{k+1}$, 
$b_{k+1} = 1 - v_{k+1}$, and $b_{k+2} = b_{k+3} = \cdots = b_d = 1$. In this case, 
$c(e,V) = d-k + ((1-v_{k+1})^{p/2} + v_{k+1}^{p/2} - 1)$, and the lemma follows. 
\ifCONF
\end{IEEEproof}
\else
\end{proof}
\fi
\hspace{5mm} We use Lemma \ref{lem:perturb} to analyze two cases. Recall that there are $B_2$ copies of the simplex $E$ in our
input set $Q$, where $B_2 = \poly(d)$ is a sufficiently large polynomial. We later specify any dependence between $B_1$ and $B_2$; so far there are none. 
\\\\
{\bf Case 1:} When we order the $\|V_{i,*}\|_2^2$ values, $v_{k+1} > \frac{3d(k+1)}{B_2}$. The intuition in this case is that the cost on the $B_2$ copies of the $d$ points in $E$ will already be too large. By Lemma \ref{lem:perturb},
we have 
\begin{eqnarray}\label{eqn:B2}
c(Q, V) \geq B_2 \cdot c(E,V) 
= B_2(d-k) + B_2 ((1-v_{k+1})^{p/2} + v_{k+1}^{p/2} - 1).
\end{eqnarray}
We analyze $x = (1-v_{k+1})^{p/2} + v_{k+1}^{p/2}$. Note that $v_{k+1} \leq 1-1/(k+1)$. Indeed, otherwise
$\|V\|_F^2 > (k+1)(1-1/(k+1)) = k$, a contradiction.  
By (\ref{eqn:ineq}), $x$ is at least
the minimum of 
$\left (\frac{1}{k+1} \right )^{p/2} + \left (1-\frac{1}{k+1} \right )^{p/2}$,
and $\left (1 - \frac{3d(k+1)}{B_2} \right )^{p/2} + \left (\frac{3d(k+1)}{B_2} \right )^{p/2}$. 
For $B_2$ a sufficiently large $\poly(d)$, again by (\ref{eqn:ineq}), $x$ is at least 
$\left (1 - \frac{3d(k+1)}{B_2} \right )^{p/2} + \left (\frac{3d(k+1)}{B_2} \right )^{p/2}$,
which in turn using that $p \leq 2$, is at least
$\left (1 - \frac{3d(k+1)}{B_2} \right ) + \left (\frac{3d(k+1)}{B_2} \right )^{p/2}$. 
Therefore, $x - 1$ is at least
$\left (\frac{3d(k+1)}{B_2} \right )^{p/2} - \frac{3d(k+1)}{B_2}$. Since $p$ is a constant
strictly less than $2$, for sufficiently large $B_2 = \poly(d)$, this is at least
$\frac{1}{2} \cdot \left (\frac{3d(k+1)}{B_2} \right )^{p/2}$. 


Plugging into (\ref{eqn:B2}), 
\begin{equation*}
c(Q,V)
	    \ge B_2(d-k) + B_2 \cdot \frac{1}{2} \cdot \left(\frac{3d(k+1)}{B_2} \right )^{p/2} 
	 = B_2(d-k) + B_2^{1-p/2} \cdot \frac{(3d(k+1))^{p/2}}{2} \geq B_2(d-k) + 2d,
\end{equation*}
where the second inequality uses that $B_2$ is a sufficiently large $\poly(d)$ and $p$ is a constant
strictly less than $2$. Inspecting (\ref{eqn:try2}), 
for $B_1 = \poly(d)$ sufficiently large (and chosen independently of $B_2$), $c(A,V) \leq d$, and therefore
for any $V' \in \mathcal{W}$, we have $c(Q,V') \leq B_2(d-k) + d$. Since $B_2 = \poly(d)$, 
the above $V$ cannot be a $(1+1/\poly(d))$-approximation to $C(Q,V^*)$. Note that there is no
dependence between $B_1$ and $B_2$ at this point, we just need each to be a sufficiently large
$\poly(d)$. 
\\\\
{\bf Case 2:} When we order the $\|V_{i,*}\|_2^2$ values, $v_{k+1} \leq \frac{3d(k+1)}{B_2}$. The intuition in this case is that the cost will be roughly the same as the case $V \in \mathcal{W}$. The condition implies that
$v_{k+2}, \ldots, v_d \leq \frac{3d(k+1)}{B_2}$. We let $S \subseteq [d]$ be the set of indices 
$i \in [d]$ corresponding to values $v_1, \ldots, v_k$. 

For an $i \in [d]$, 
we write $A_i = A_i^S + A_i^{[d] \setminus S}$, where $A_i^S$ is $0$ outside of the columns in $S$, and 
$A_i^{[d]\setminus S}$ is $0$ outside of the columns of $[d]\setminus S$. Similarly we write 
$V = V^S + V^{[d] \setminus S}$, where $V^S$ is $0$ outside of the rows in $S$, and 
$V^{[d] \setminus S}$ is $0$ outside of the rows of $[d] \setminus S$. Then
$A_iV = A_i^S V^S + A_i^{[d] \setminus S} V^{[d] \setminus S}$, and so
\begin{eqnarray*}
\|A_i V\|_2 & \leq & \|A_i^S V^S\|_2 + \|A_i^{[d] \setminus S} V^{[d] \setminus S}\|_2
\leq \|A_i^S\|_2 + \|A_i^{[d] \setminus S}\|_2 \|V^{[d] \setminus S}\|_2\\
& \leq & \|A_i^S\|_2 + \|V^{[d] \setminus S}\|_F
 \leq \|A_i^S\|_2 + \left (\frac{3d(k+1)(d-k)}{B_2} \right )^{1/2},
\end{eqnarray*}
where the first inequality is the triangle inequality, 
the second inequality uses sub-multiplicativity of the operator norm,
the third inequality uses that $\|A_i\|_2 = 1$ and that the Frobenius norm upper bounds the operator norm,
and the final inequality uses that $\|V^{[d] \setminus S}\|_F = \left (\sum_{i=k+1}^d v_{k+1} \right )^{1/2}$. 
This implies 
$\|A_i V\|_2^2 \leq \|A_i^S\|_2^2 + 3\left (\frac{3d(k+1)(d-k)}{B_2} \right )^{1/2},$
using that $\|A_i^S\|_2 \leq 1$ and $B_2$ is a sufficiently large $\poly(d)$. Hence,
\begin{eqnarray}\label{eqn:continue}
d(A_i, V)^p & = & (1 - \|A_i V\|_2^2)^{p/2} 
\geq \left (1 - \|A_i^S\|_2^2 - 3 \left (\frac{3d(k+1)(d-k)}{B_2} \right )^{1/2} \right )^{p/2}.
\end{eqnarray}
If $i \in S$, then $\|A_i^S\|_2^2 = (1-1/B_1)^2 + e(i,S)c^2/(B_1r)$ by (\ref{eqn:inS}). In this
case, plugging into (\ref{eqn:continue}) and taking out a $1/B_1^{p/2}$ factor,  
$$d(A_i, V)^p \geq (1/B_1^{p/2}) \left (2 - e(i,S)c^2/r - 1/B_1-  3 \left (\frac{3d(k+1)(d-k)B_1^2}{B_2} \right )^{1/2} \right )^{p/2}.$$
Now we introduce a dependence between $B_2$ and $B_1$: by 
making $B_2$ a sufficiently large $\poly(d)$ factor larger than $B_1$, we can absorb
the $3 \left (\frac{3d(k+1)(d-k)B_1^2}{B_2} \right )^{1/2}$ term into the $1/B_1$ term, obtaining:
$$d(A_i, V)^p \geq (1/B_1^{p/2})(2 - e(i,S)c^2/r - O(1/B_1))^{p/2}.$$
On the other hand, if $i \notin S$, then $\|A_i^S\|_2^2 = e(i,S) c^2/(B_1r)$ by (\ref{eqn:notinS}). In this case, 
by (\ref{eqn:continue}), 
\begin{eqnarray*}
d(A_i, V)^p & \geq & \left (1-e(i,S)c^2/(B_1 r) - 3 \left (\frac{3d(k+1)(d-k)}{B_2} \right )^{1/2} \right )^{p/2} \geq 1-O(1/B_1),
\end{eqnarray*}
where the second inequality follows by again making $B_2$ a sufficiently large $\poly(d)$ factor larger than $B_1$. We thus have as
a lower bound the expression in (\ref{eqn:tryAgain}), that is, 
$c(A,V) \geq d- |S| -O(d/B_1) + (1/B_1^{p/2}) \sum_{i \in S} (2 - e(i,S) c^2/r - O(1/B_1))^{p/2}.$
We obtain the same conclusion as in (\ref{eqn:cliqueNotThere}), 
as all that has changed is the absolute constant in the $O(d/B_1)$ notation, and the $1/B_1$ term in each summand is now $O(1/B_1)$, which does not affect the bound in 
(\ref{eqn:cliqueNotThere}), other than changing the constant in the additive $O(d^2/B_1)$ term. Moreover, $c(E,V) \geq d-k$ by Lemma \ref{lem:simplex}, and so even if $V \notin \mathcal{W}$, if there is no clique of size at least $k$, the cost $c(Q,V)$ is an additive
$\Omega((1/B_1)^{p/2}/r^2)$ factor larger than the optimal $V$ (which is possibly in $\mathcal{W}$) when there is a clique of size at least $k$.  

\hspace{5mm} {\bf Wrapup:} It follows that $c(Q, V^*)$ is an additive $\Omega((1/B_1)^{p/2}/r^2)$ factor larger if 
there is no clique of size $k$, versus if there is a clique of size $k$. 
Since $c(Q,V) \leq B_2(d-k) + c(A,V) \leq B_2(d-k) + d$, using that $\|A_i\|_2^p = 1$ for all $i$, it follows that any 
algorithm which outputs a $k$-dimensional subspace $V$ for which $c(Q, V) \leq (1+1/\poly(d))c(Q, V^*)$, for a large enough
$\poly(d)$, can be used to solve the {\sf Clique} problem on $r$-regular graphs on $d$ vertices, in which the goal is to decide
if there is a clique of size at least $k$. 
The theorem now follows from the NP-hardness of the {\sf Clique} problem. 
%
\ifCONF
\end{IEEEproof}
\else
\end{proof}
\fi

\ifCONF\else

\section{Regression with $M$-estimators}\label{sec regress}

We present our proof of Theorem~\ref{thm regress}.


\begin{proof}
Let $\Delta^*$ denote the cost of the optimal solution.

The procedure is recursive, for a constant number of levels.
For $A$ at a given level, let $\hat A\equiv [A\ b]$.
At a general level of recursion, a weighted version of the problem is to be solved.
Compute leverage scores of $\hat A$ as in Lemma~\ref{lem M lev}, for each $U^j$
applying
Theorem~\ref{thm well cond} for $p=2$ and $m=d$,
obtaining a change-of-basis matrix $R^j$ in $O(\nnz(U^j)+\poly(d))$ time,
and then applying Theorem~\ref{thm G est}, to estimate the Euclidean row norms
of the $U^j$. From that theorem, and applying Lemma~\ref{lem M kinda embed}
with $z=1$ and $\rho=10\Delta^*$,
we have that for sample size
\[
O(n^{1/2+\kappa})\poly(d)\log(1/\delta)/\eps^2
\]
for fixed small $\kappa>0$, every $x\in\R^{d+1}$ has
either $\norm{S\hat A x}_v\ge 10\Delta^*$,
or $\norm{S\hat A x}_v = \norm{\hat A x}_v(1\pm\eps) \pm\eps\Delta^*$.
This implies that an $\eps$-approximate solution to $\min_{x\in\R^d}\norm{S(Ax-b)}_v$
is an $O(\eps)$-approximate solution to the original problem.
The inputs to the next level of recursion are $SA$, $Sb$, and the weights $w$
determining $\norm{S\cdot}_v$.

After a constant number of recursive steps, with constant blowup in error
and failure probability and in the weights, the resulting matrix $A'$ has $m=n^\beta\poly(d/\eps)$ rows,
for some $\beta>0$.  Pick $\beta$ such that the ellipsoid method can
be applied, with a running time $O(m^C)$ for a constant $C<1/2\beta$,
taking time $(n^\beta\poly(d/\eps))^C <\sqrt{n}\poly(d/\eps) =O(n)+\poly(d/\eps)$.
The theorem follows.
\end{proof}
\fi 

\section*{Acknowledgments}

We acknowledge the support of the XDATA program
of the Defense Advanced Research Projects Agency (DARPA), administered through Air Force Research Laboratory contract FA8750-12-C-0323.
We are very grateful to Danny Feldman, Nimrod Megiddo, Kasturi Varadarajan, and Nisheeth Vishnoi for many helpful discussions.

\ifCONF
\else
\bibliographystyle{plain}
\bibliography{main}
\newpage
\fi

\ifCONF\else
\appendix
\section{$L_1-L_2$ Subadditivity}\label{app:l1l2loss}
\begin{theorem}
For $L_1-L_2$ loss function $f(x) = 2(\sqrt{1+x^2/2}-1)$, its square-root
$g(x) = \sqrt{2} (\sqrt{1+x^2/2}-1)^{1/2}$ is a subadditive function. 
\end{theorem}
\begin{proof}
Since $g(0) = 0$, 
by Remark 2.19 of 
\cite{M11}
, it suffices to show $g(x)/x$ is a decreasing function
for $x \in \mathbb{R}^{\geq 0}$. The function $g(x)/x$ is decreasing
if and only if the function $h(x) = (\sqrt{1+x^2/2}-1)^{1/2}/x$ is decreasing.

Plugging $h(x)$ into an online derivative calculator, we have
$$h'(x) =  \frac{1}{4 \sqrt{1+x^2/2} \sqrt{\sqrt{1+x^2/2}-1}}
- \frac{\sqrt{\sqrt{x^2/2+1}-1}}{x^2},$$
and it suffices to show $h'(x) < 0$, which is equivalent to showing
$$x^2 \leq 4 \sqrt{x^2/2+1} (\sqrt{x^2/2+1}-1),$$
or 
$$x^2 \leq 2x^2+4 - 4\sqrt{x^2/2+1},$$
or 
$$16(x^2/2+1) \leq x^4 + 8x^2 + 16,$$
which is equivalent to $0 \leq x^4$, which holds. 
\end{proof}

\fi

\end{document}